\documentclass[journal]{IEEEtran}
\ifCLASSINFOpdf
\else
\fi
\usepackage[utf8]{inputenc}
\DeclareUnicodeCharacter{2212}{-}%
\usepackage{amsmath}
\usepackage{amsfonts}
\usepackage{amssymb}
\usepackage{graphicx}
\usepackage[numbers,sort&compress]{natbib}
\usepackage{subcaption}
\usepackage{caption}
\usepackage[dvipsnames]{xcolor}
\usepackage{algorithmicx}
\usepackage{algorithm}
\usepackage{algpseudocode}
\usepackage{comment}
\usepackage{bm}
\usepackage{amsthm}

\usepackage{autonum}

\newcommand{\y}{\boldsymbol{y}}
\newcommand{\x}{\boldsymbol{x}}
\newcommand{\n}{\boldsymbol{n}}
\newcommand{\0}{\boldsymbol{0}}
\newcommand{\muv}{\boldsymbol{\mu}}
\newcommand{\dv}{\boldsymbol{d}}
\newcommand{\Dv}{\boldsymbol{D}}
\newcommand{\Hv}{\boldsymbol{H}}

\newcommand{\Iv}{\boldsymbol{I}}
\newcommand{\Qv}{\boldsymbol{Q}}
\newcommand{\Cv}{\boldsymbol{C}}
\newcommand{\Sigmav}{\boldsymbol{\Sigma}}
\newcommand{\alphav}{\boldsymbol{\alpha}}
\newcommand{\betav}{\boldsymbol{\beta}}
\newcommand{\nuv}{\boldsymbol{\nu}}

\DeclareMathOperator*{\argmin}{arg\,min}
\DeclareMathOperator*{\argmax}{arg\,max}

\newtheorem{theorem}{Theorem}
\newtheorem{lemma}[theorem]{Lemma}
\newtheorem{proposition}[theorem]{Proposition}
\theoremstyle{definition}
\newtheorem{definition}{Definition}

\begin{document}
%
\title{MMSE Approximation For Sparse Coding\\ Algorithms Using Stochastic Resonance}
%
%
%

\author{Dror Simon,
  Jeremias Sulam,
  Yaniv Romano,
  Yue M. Lu
  and Michael Elad

\thanks{D. Simon and M. Elad are with the Department of Computer Science at the Technion, Israel.}
\thanks{J. Sulam is with the Biomedical Engineering Department and the Mathematical Institute of Data Science of Johns Hopkins University, USA.}
\thanks{Y. Romano is with the Statistics Department of Stanford University, USA.}
\thanks{Y. M. Lu is with the John A. Paulson School of Engineering and Applied Sciences, Harvard University, USA.}
}

\maketitle

\begin{abstract}
Sparse coding refers to the pursuit of the sparsest representation of a signal in a typically overcomplete dictionary. From a Bayesian perspective, sparse coding provides a Maximum a Posteriori (MAP) estimate of the unknown vector under a sparse prior. In this work, we suggest enhancing the performance of sparse coding algorithms by a deliberate and controlled contamination of the input with random noise, a phenomenon known as stochastic resonance. The proposed method adds controlled noise to the input and estimates a sparse representation from the perturbed signal. A set of such solutions is then obtained by projecting the original input signal onto the recovered set of supports. We present two variants of the described method, which differ in their final step. The first is a provably convergent approximation to the Minimum Mean Square Error (MMSE) estimator, relying on the generative model and applying a weighted average over the recovered solutions. The second is a relaxed variant of the former that simply applies an empirical mean. We show that both methods provide a computationally efficient approximation to the MMSE estimator, which is typically intractable to compute. We demonstrate our findings empirically and provide a theoretical analysis of our method under several different cases.
\end{abstract}


%
\IEEEpeerreviewmaketitle

\section{Introduction}
\label{sec:introduction}
%
%
%
%


\IEEEPARstart{I}{n} signal processing, often times we have access to a corrupted signal and we wish to estimate its clean version. This process includes a wide variety of problems, such as denoising, where we wish to remove noise from a noisy signal; deblurring where we look to sharpen an image that has been blurred or was taken out of focus; and inpainting in which we fill-in missing data that have been removed from the image.
All the aforementioned tasks, and many others, include a linear degradation operator and a stochastic corruption. The forward model can be described by $\y=\Hv\x+\nuv$, where $\x$ is the clean signal, $\Hv$ is the linear degradation operator, $\nuv$ denotes additive noise, and $\y$ stands for the noisy measurements.

In order to provide a \emph{good estimate} of $\x$, it is useful to incorporate both the statistical properties of the corruption as well as prior knowledge on the signal. In image processing in particular, many such priors have been developed over the years, such as total-variation, self-similarity, sparsity, and many others \cite{rudin1992nonlinear,Dabov,Elad2006}. 

In this work we focus our attention on the sparse model prior, which assumes that the clean signal $\x\in \mathbb{R}^{n}$ is a linear combination of a small number of columns from an overcomplete dictionary $\Dv\in \mathbb{R}^{n\times m}$, where $n < m$, referred to as \emph{atoms}. In this case, we can write $\x=\Dv\alphav$, where the representation vector $\alphav\in\mathbb{R}^m$ is sparse. 
One of the most fundamental problems in this model is termed \emph{sparse coding}: Given $\x$, find the sparsest $\alphav$ such that $\x=\Dv\alphav$.  Formally, this calls for solving
\begin{equation}
\left(P_{0}\right): \quad \hat{\alphav} = \argmin_{\alphav} ~ \left|\left|\alphav\right|\right|_{0} ~ \text{s.t.} ~ \Dv\alphav=\x,
\end{equation}
where $\|\cdot\|_0$ stands for the $l_0$ pseudo-norm that counts the number of non-zero elements in the vector.

Returning to the real measurements setting, and in particular for a denoising task, the degradation is simply given by additive noise $\nuv$, typically assumed Gaussian or with bounded energy $\|\nuv\|_{2} \leq \epsilon$, resulting in $\y=\x+\nuv$. Hence, the above problem is naturally modified to 
\begin{equation}
\left(P_{0}^{\epsilon}\right): \quad \hat{\alphav} = \argmin_{\alphav}~ \left|\left|\alphav\right|\right|_{0} ~\text{s.t.} ~ \left|\left|\Dv\alphav-\y\right|\right|_{2} \leq \epsilon.
\label{eq:P0epsilon}
\end{equation}
The solution of $(P_{0}^{\epsilon})$ can be then used to provide an estimate of $\x$, in the form of $\hat{\x}=\Dv\hat{\alphav}$.

Without further assumptions, $\left(P_{0}^{\epsilon}\right)$ is non-convex and NP-Hard in general \cite{natarajan1995sparse}, as it requires searching through all the possible supports of $\alphav$. Nonetheless, different approximation or pursuit algorithms have been developed in order to manage this task effectively. Some of these include greedy strategies, such as the  \emph{Orthogonal Matching Pursuit} (OMP) \cite{Tropp2007}, or relaxation alternatives, like \emph{Basis-Pursuit} (BP) \cite{chen2001atomic}.

These approximation algorithms have been accompanied by theoretical guarantees for finding a sparse representation $\hat{\alphav}$ that is close to the original representation, e.g. in an $\ell_2$ sense, $\|\hat{\alphav}-\alphav\|_{2}$. Additionally, they often assure the correct recovery of the support \cite{donoho2006stable}. Such results rely on the cardinality of $\alphav$, the range of the non-zero values and properties of the dictionary $\Dv$. These  algorithms succeed not only in cases of noise with bounded energy, but also accurate solutions with high probability in more general settings \cite{Candes2004,tropp2005average}.

From a Bayesian point of view, pursuit algorithms provide an approximation to a Maximum a Posteriori (MAP) estimator \cite{schniter2008fast,Elad2009}. Indeed, the objective seeks the most likely signal under the sparse prior, subject to the noise deviation. Such an estimator does not coincide with the Minimum Mean Squared Error (MMSE) estimate, which is of great interest in many cases. Unfortunately, however, exact MMSE estimation has been shown to be computationally intractable and unfeasible in practice \cite{schniter2008fast,Elad2009}.

In this paper, we provide an efficient way to approximate the MMSE by means of contaminating the input signal with a controlled amount of (further) noise. The proposed method leans on the \emph{Stochastic Resonance} (SR) phenomenon, in which the addition of noise to a weak signal can increase its output Signal-to-Noise Ratio (SNR) in the context of a non-linear transfer function. This field has been broadly developed to improve the performance of sub-optimal detectors \cite{Kay2006}, non-linear parametric estimators \cite{Chen2008} and image processing algorithms \cite{Chen2014}. As we will careful comment later, our method shares similarities with the Supra-threshold SR (SSR) algorithm \cite{stocks2000suprathreshold} while providing a generalization thereof. After providing a provably convergent algorithm, we will explore an alternative relaxation that will enable the deployment of this idea to general pursuit methods and in more practical scenarios. The proposed approach provides a general tool that can improve performance of different sparse coding algorithms, both in synthetic and real data settings, as we demonstrate numerically.  

The rest of the paper is structured as follows. Section \ref{sec:related} explores and comments on related previous work, and Section \ref{sec:bayes} reviews Bayesian estimation under the sparse prior. Then, in Section \ref{sec:prior_algorithm} we present and analyze our provably convergent MMSE approximation algorithm, followed by a practical variation in Section \ref{sec:general_algorithm}. We study the properties of the general algorithm under specific cases in Sections \ref{sec:unitaryPursuit} and \ref{sec:singleAtom}. In Section \ref{sec:noise} we explore the application of different SR noise distributions, before showcasing our proposed approach for image denoising in Section \ref{sec:image}. Finally, we conclude our work in Section \ref{sec:conclusion}.


\section{Related Work}
\label{sec:related}

Providing solutions to inverse problems with minimal MSE has remained a problem of great interest for many years. In the context of sparse modeling in particular, the work in \cite{Elad2009} suggested an MMSE approximation in terms of the \emph{Random-OMP} (RandOMP) algorithm. This approach consists in running a stochastic variant of OMP several times, introducing some randomness in the choice of the supports at every iteration, and finally averaging the results. RandOMP was shown to coincide with the MMSE estimator under a unitary dictionary assumption, or when $\Dv$ is overcomplete and the cardinality of the representations is restricted to one atom. RandOMP also improves the MSE empirically in more general cases, where the MMSE cannot be practically computed. On the other hand, this approach inherits the limitation of OMP, and specifically the gradual increase of the support one element at a time. As a result, this method is impractical in cases where the cardinality of the solution is in the order of hundreds and beyond -- as for many real world signals. 

In \cite{schniter2008fast} a pursuit based on a Bayesian approach was suggested. The \emph{Fast Bayesian Matching Pursuit} (FBMP) is a method that seeks the most probable supports and then approximates their posterior probabilities in order to provide an estimate of the MMSE. FBMP has been developed under a specific sparse prior model and relies on it in order to properly compute its estimate. This limits this approach to cases that follow the assumed signal model and less applicable to real applications where the prior is unknown.

A closely related method to our work is that of Supra-threshold Stochastic Resonance, first described in \cite{stocks2000suprathreshold}. SSR consists in the addition of noise to a signal before it is passed through a set of thresholds, or other analytic non-linearities \cite{chapeau2006noise}. All outputs are then averaged to obtain a final result. Just as in SR, the amount of noise to be added is a parameter that needs careful tuning, and it can be set to maximize some statistical measure (e.g. SNR, Mutual Information, among others). On the other hand, SSR is usually motivated by a fixed physical system (e.g. sensory neurons \cite{stocks2001generic}) where one has only limited control over the input signal. 

As we will explore in the following sections, our work is inspired by SSR and it can be understood as a generalization of it. In particular, we will regard pursuit algorithms as more general non-linear functions, moving beyond simple thresholding rules. Moreover, the output of our algorithm will subtract some of the effect of the added noise, providing a better estimate, asymptotically converging to the MMSE estimator. Finally, the proposed approach will be general, making it possible to consider convex relaxation alternatives, this way making MMSE approximation plausible for large dimensional signals. Before moving to the presentation of the main algorithm, however, we review some general results in Bayesian sparse estimation in the following Section.



\section{Bayesian Estimation Under the Sparse Prior}
\label{sec:bayes}
Let us formulate our model in more details. The matrix $\Dv \in \mathbb{R}^{n\times m}$ is an overcomplete dictionary, while the representation $\alphav\in \mathbb{R}^{m}$ is a sparse vector with either fixed cardinality $\|\alphav\|_0=M$ or with a prior probability $p_i$ for each entry to be non-zero -- we will use both alternatives in this work. The generative model consists of first drawing a support $S$ from a distribution $P_S$. We denote the set of all  probable supports by ${\Omega=\left\{S|P_S(S)>0\right\}}$. The non-zero elements in $\alphav$, denoted as $\alphav_S$, are then drawn from a distribution $P_{\alpha_S}$, which is assumed to be a white Gaussian distribution $\mathcal{N}(\bm{0},\sigma_{\alpha}^2\Iv)$. A signal $\x$ is constructed by a linear combination of atoms $\x=\Dv\alphav$, and the measured samples are given by $\y= \x + \nuv $, where $\nuv \sim \mathcal{N}(\bm{0},\sigma_{\nu}^2\Iv)$ is additive Gaussian noise. Under such a generative model, a Bayesian formulation for estimating $\alphav$ deploys the prior on the representations $\alphav$ in different ways. We now describe different Bayesian estimators that can be formulated in this context, as described in \cite{Turek2011}. 

We begin with the Oracle estimator, which seeks to estimate the clean signal when the support is assumed to be known. This is a simplistic assumption, as retrieving the correct support $S$ of the original sparse representation $\alphav$ is the essence of the combinatorial $\left(P_{0}^{\epsilon}\right)$ problem. In such a case, the MMSE estimator is simply the conditional expectation $\hat{\alphav}_S^{\text{Oracle}} = \mathbb{E}\left[\alphav|\y, S\right]$, given by
\begin{equation}
\hat{\alphav}^{\text{Oracle}}_{S}(\y) = \frac{1}{\sigma_\nu^2}\Qv_S^{-1}\Dv_S^T\y,
\label{eq:oracle}
\end{equation}
where $\Dv_S$ is the sub-dictionary containing only the atoms indexed by $S$, and $\Qv_S$ is given by
\begin{equation}
\Qv_S=\frac{1}{\sigma_\alpha^2}\Iv_{|S|} + \frac{1}{\sigma_\nu^2}\Dv_S^T \Dv_S.
\end{equation}
We refer to this estimator as the \emph{Oracle}, as there is no possible way of knowing the true support beforehand.

Next is the MAP estimator, which searches for the most probable support $\hat{S}$ given the measurements and uses it to estimate the signal\footnote{In fact, this is the MAP of the support. We use it to avoid the probable case where the recovered signal is the $\0$ vector as described in \cite{Turek2011}.}. The relevant posterior probability for this estimation is given by \cite{Turek2011}
\begin{eqnarray} 
P(S|y)=\frac{t_S}{t}, \quad\quad\quad t\triangleq \sum_{S\in\Omega} t_S
\label{eq:mmseGeneral}
\end{eqnarray}
where 
\begin{eqnarray}
t_S \triangleq \frac{1}{\sqrt{\det(\Cv_S)}}\exp{\left\{ -\frac{1}{2}\y^T\Cv_S^{-1}\y \right\}}\prod_{i \in S}p_i \prod_{i \notin S}\left(1-p_i\right), \nonumber
\end{eqnarray}
and $\Cv_S^{-1}=\frac{1}{\sigma_{\nu}^2}\Iv_n - \frac{1}{\sigma_{\nu}^4}\Dv_S\Qv_S^{-1}\Dv_S^T$.

The MAP estimator is obtained by maximizing $P(S|\y)$ with respect to $S$, which (employing Bayes' rule) can be written as
\begin{align}
\hat{S}=\argmax_S P(S|\y) = \argmax_S P(\y|S) P(S).
\end{align}
By replacing the conditional probability and the prior on the support, one can show \cite{Turek2011} that this corresponds to
\begin{align}\label{eq:mapS}
\hat{S}=\argmax_S  \frac{1}{2}\left|\left| \frac{1}{\sigma_v^2}\Qv_S^{-\frac{1}{2}}\Dv_S^T\y \right|\right|_2^2 - \frac{1}{2}\log(\det(\Cv_S)) \\+ \sum_{i \in S} \log(p_i) + \sum_{j \notin S} \log(1-p_j).
\end{align}
In the case where the cardinality of $\alphav$ is constant, $\|\alphav\|_0=M$, and all supports are equally likely, the last two terms above can be omitted. Once the most probable support is recovered, the oracle formula can then be employed to estimate the corresponding coefficients as
\begin{equation}
\hat{\alphav}^\text{MAP}(\y) = \hat{\alphav}^{\text{Oracle}}_{\hat{S}^{\text{MAP}}}(\y).
\end{equation}

The last estimator we discuss here is the MMSE, which is given by the conditional expectation,
\begin{equation}
\hat{\alphav}^{\text{MMSE}}(\y) = \mathbb{E}\left[ \alphav \middle| \y \right] = \sum_{S\in\Omega} P(S|\y) \mathbb{E}\left[ \alphav \middle| \y,S \right].
\end{equation}
This is a weighted sum over all the possible supports. Moreover, each element calls for the oracle estimator over the candidate support, allowing us to write
\begin{equation}
\hat{\alphav}^{\text{MMSE}}(\y) = \sum_{S\in\Omega} P(S|\y) \hat{\alphav}^{\text{Oracle}}_{S}(\y).
\label{eq:mmseSum}
\end{equation}
Because of this massive averaging of all the possible supports, somewhat surprisingly, the MMSE under a sparse prior is actually a dense vector. 

Both estimators, $\hat{\alphav}^{\text{MMSE}}$ and $\hat{\alphav}^{\text{MAP}}$, are NP hard to obtain in general, as they require either a sum over all the possible supports or the computation of all posterior probabilities and selecting the highest one. Either option is prohibitive as the number of possible supports is exponentially large. This is the reason for approximation algorithms or pursuits, which typically attempt to approximate the MAP estimate.
Nonetheless, as noted in \cite{schniter2008fast}, the posterior probabilities $P(S|\y)$ have an exponential nature -- see Equation \eqref{eq:mmseGeneral}, and thus the sum in Equation \eqref{eq:mmseSum} is practically dominated by a much smaller number of terms. Put formally, this suggests that there exists a subset of the supports, $\omega \subset \Omega,~ |\omega| \ll |\Omega|$ such that $P(S_{\omega}|\y)\gg P(S_{\Omega\setminus\omega}|\y)$ where $S_{\omega}\in\omega,S_{\Omega\setminus\omega}\in\Omega\setminus\omega$.
If one could obtain these most significant elements and their proper weights, then an MMSE approximation would be attainable. This is the rationale in earlier work on MMSE approximations \cite{schniter2008fast,Elad2009}, and the motivation behind our proposed approach.

\section{The Proposed Algorithm}
\label{sec:prior_algorithm}

\begin{algorithm}[t]
\caption{Prior-based SR algorithm}\label{algorithm:SR_importance}
\begin{algorithmic}[1]
\Procedure{PriorBased-SR}{$\y, \Dv$, PursuitAlg, $\sigma_n$, $K$}
\State $\bm{\mathcal{S}} \gets \Phi$
\For{k=1...K}
\State $\n_k\gets$ SampleNoise($\sigma_n$)
\State $\tilde{\alphav}_k\gets$ PursuitAlg$(\y+\n_k, \Dv)$
\State $\hat{S}_k \gets $ Support$(\tilde{\alphav}_k)$
\State $\bm{\mathcal{S}} \gets \bm{\mathcal{S}} \cup \{\hat{S}_k\}$
\EndFor
\State $\hat{\alphav} \gets \frac{\sum_{S\in\bm{\mathcal{S}}}P(\y|S)P(S)\hat{\alphav}_{S}^{\text{Oracle}}(\y)}{\sum_{S\in\bm{\mathcal{S}}}P(\y|S)P(S)}$
\State \textbf{return} $\hat{\alphav}$
\EndProcedure
\end{algorithmic}
\end{algorithm}

We now present the proposed approach in Algorithm \ref{algorithm:SR_importance}. This algorithm consists of $K$ iterations, where in each a small amount of noise $\n_k$ is added to the already noisy signal $\y$, and a (greedy or relaxation-based) pursuit algorithm is employed. The final estimation is computed as a weighted average over the obtained supports: For each recovered support, the corresponding oracle estimator is obtained w.r.t. the original measurements $\y$ (i.e., without the influence of $\n_k$) and weighed according to its un-normalized posterior probability.

As we show next, Algorithm \ref{algorithm:SR_importance} asymptotically converges to the MMSE estimator. Before presenting this result, however, we present in the following lemma an alternative expression for the MMSE estimator that will prove useful later on.


\begin{lemma}
\label{lemma:mmseSum}
Let $\Omega=\left\{S\middle|P(S)>0\right\}$ be the set of all the possible supports, then
\begin{equation}
 {\hat{\alphav}^{\text{MMSE}}(\y) = \frac{\sum_{S\in\Omega}P(\y|S)P(S)\hat{\alphav}_{S}^{\text{Oracle}}(\y)}{\sum_{S\in\Omega}P(\y|S)P(S)}}.
\end{equation}
\end{lemma}
\begin{proof}
From Bayes' theorem and the law of total probability, we can write
\begin{equation}
    P(S|\y) = \frac{P(\y|S)P(S)}{P(\y)} = \frac{P(\y|S)P(S)}{\sum_{S\in\Omega}P(\y|S)P(S)}.
\end{equation}
Combining this with Equation \eqref{eq:mmseSum}, we have:
\begin{equation}
\hat{\alphav}^{\text{MMSE}}(\y) = \frac{\sum_{S\in\Omega}P(\y|S)P(S)\hat{\alphav}_{S}^{\text{Oracle}}(\y)}{\sum_{S\in\Omega}P(\y|S)P(S)}.
\end{equation}
\end{proof}

The result that we present below is general, in the sense that it is applicable to any pursuit that provides a stable approximation of $\hat{\alphav}$. For this reason, we now formalize this in the following definition, followed by the statement of the main theorem. 

\begin{definition}
A pursuit method $PM$ is a \emph{stable pursuit} w.r.t. the prior $P_S,~P_{\alphav_S}$ and the dictionary $\Dv$ if $\forall~ S\in\Omega$ and a respective $\alphav\in\mathbb{R}^m$ such that ${\text{Support}(\alphav)=S}$, and ${P_{\alphav_S}(\alphav_S)>0}$, then $\exists\epsilon>0$ such that ${\forall \bm{v}\in \left\{\bm{v}\in\mathbb{R}^n \middle| \left\|\bm{v}-\Dv\alphav\right\|_2 < \epsilon \right\}}$ the pursuit method $PM$ recovers the correct support, i.e. ${\text{Support}\left\{PM\left(\bm{v}\right)\right\}= S}$.
\end{definition}

\begin{theorem}
\label{theorem:importance_sampling}
Let $\n_k \sim \mathcal{N}\left(\bm{0},\sigma_n^2\Iv\right)$ be a white Gaussian SR noise, $\sigma_n>0$ and $PM$ a stable pursuit method. Then, as ${K\to\infty}$, Algorithm \ref{algorithm:SR_importance} asymptotically converges to the MMSE estimator with probability $1$.
\end{theorem}
 

\begin{proof}
Assume by contradiction that Algorithm \ref{algorithm:SR_importance} does not converge to the MMSE. From Lemma \ref{lemma:mmseSum} this means that ${\Omega \setminus \bm{\mathcal{S}}\neq\Phi}$, where $\bm{\mathcal{S}}$ are the gathered supports by Algorithm \ref{algorithm:SR_importance}. Let $S_i$ be a support such that $S_i\in \Omega \setminus \bm{\mathcal{S}}$, and let $\alphav_i\in\mathbb{R}^m$ be such that $\text{Support}\{\alphav_i\}=S_i$ and ${P_{\alphav_S}(\alphav_{S_i})>0}$. Since the pursuit method is stable, $\exists \epsilon>0$ such that ${\forall\bm{v}\in\left\{\bm{v}\in\mathbb{R}^n \middle| \|\bm{v}-\Dv\alphav_i\|_2<\epsilon\right\}}$ and ${\text{Support}\left\{PM\left(\bm{v}\right)\right\}=S_i}$. 

In each iteration, the algorithm sparse codes $\y + \n_k$ and since $\n_k$ is Gaussian ${P(\|\y + \n_k - \Dv\alphav\|_2<\epsilon) > 0}$ for any $\alphav\in\mathbb{R}^m$. Hence, as $K\to\infty$, $\exists k_i$ such that ${\|\y + \n_{k_i} - \Dv\alphav_i\|_2<\epsilon}$ and from the stability of the pursuit $\text{Support}\left\{PM\left(\y + \n_{k_i}\right)\right\}=S_i$. Therefore, at the $k_i$-th iteration, the support $S_i$ is added to the accumulated set $\bm{\mathcal{S}}$, contradicting the false assumption.
\end{proof}

A few remarks are in place. First, if $K\geq |\Omega|$, using Algorithm \ref{algorithm:SR_importance} is clearly ineffective since one may simply compute the MMSE from \eqref{eq:mmseSum}. Nevertheless, fast convergence occurs when $\bm{\mathcal{S}}$ contains the most likely supports, since their weight is much larger than the weight of the other elements \cite{schniter2008fast}. Using the MAP estimator -- or its approximation in terms of pursuit algorithms -- promotes highly likely supports more often than less likely ones. The SR idea provides a way of accumulating a set of highly probable supports by perturbing the measurements before running the pursuit.

Second, the result above is somewhat limited as it does not inform about \emph{how fast} $\hat{\alphav}$ converges to $\hat{\alphav}^{\text{MMSE}}$. Such an analysis must depend on the energy of the added noise. Indeed, when $\sigma_n$ is too large, the SR measurements $\y + \n_k$ significantly deviate from $\y$, reducing the chances to retrieve probable supports. On the other hand, when the added noise is weak, the signal $\y + \n_k$ hardly varies, reducing the chances to cover the set of supports quickly. The analysis of this convergence rate is challenging, and we defer it to future work. As we will extensively corroborate numerically, however, a relatively low number of iterations $K$ suffices to provide a good approximation to the MMSE in practice.

To empirically examine the performance of Algorithm \ref{algorithm:SR_importance}, we performed the following experiments. We drew a random dictionary $\Dv\in\mathbb{R}^{50\times 100}$ and normalized its columns. Then, we generated sparse vectors containing a single non-zero element at a random location, where its value was sampled from the normal distribution $\mathcal{N}(0,1)$. Signals were then created by multiplying the sparse vectors with the dictionary $\Dv$. Finally, we added a Gaussian noise $\nuv\sim\mathcal{N}(\bm{0},\sigma_{\nu}^2\Iv)$, $\sigma_{\nu}=0.2$ to the signals resulting with noisy measurements. We then used Algorithm \ref{algorithm:SR_importance} to estimate the clean signals and compared its MSE to the MAP and the MMSE estimators for a varying number of iterations $K$ and standard deviation $\sigma_{n}$. Note that in this case, the OMP algorithm is also the MAP, and therefore we used it as our pursuit method. The results averaged over $10,000$ realizations can be seen in Figure \ref{fig:importance1}. Note that in this experiment, the number of possible supports is $|\Omega|=100$. When $K=100$ and $\sigma_n\approx0.5$, the MSE of Algorithm \ref{algorithm:SR_importance} almost reaches that of the MMSE estimator, suggesting that only a few improbable supports are missing from the accumulated set $\bm{\mathcal{S}}$. That said, even when $K$ is much smaller than $100$, for a reasonable amount of SR noise $\sigma_{n}\approx\sigma_{\nu}$, the MSE of Algorithm \ref{algorithm:SR_importance} is close to that of the MMSE estimator.

To further demonstrate the efficiency of this technique, we repeat the same protocol, only this time the number of non-zero elements in the sparse vector is 3, increasing the number of possible supports to $|\Omega|=\binom{100}{3} = 161,700$, all apriori equally likely. The results can be seen in Figure \ref{fig:importance2}. Note that while the MAP and the MMSE estimators require iterating through all the possible supports, Algorithm \ref{algorithm:SR_importance} efficiently retrieve the most significant ones, leading to superior denoising results over the standard OMP and the MAP estimator.

These results show that not only does the algorithm asymptotically converge to the MMSE, but also a significant improvement can be achieved over the MAP estimator (or its approximation) even with a relatively small number of iterations. A limitation of this approach, however, is that this method requires full knowledge of the generative model and its parameters, similar to the FBMP algorithm, thus limiting its use in real settings. In the following sections we suggest a practical variation of the presented algorithm that can be used in more general cases.

\begin{figure}[t]
	\centering
	\includegraphics{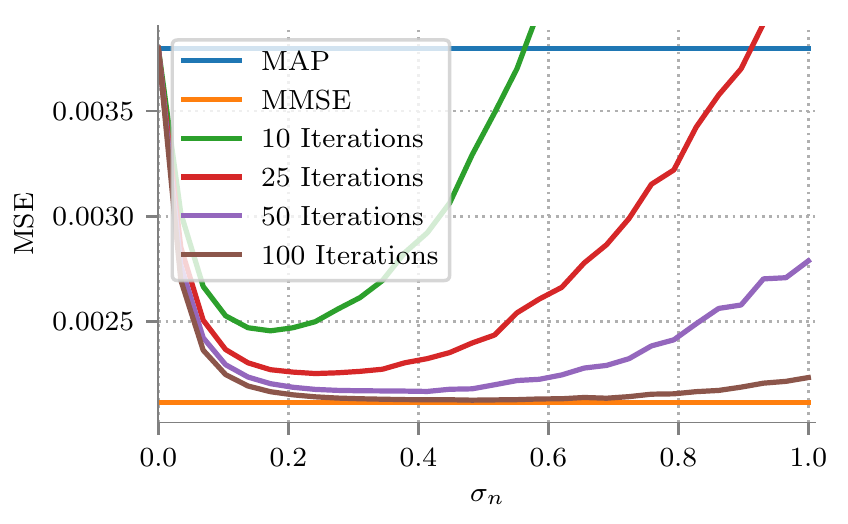}
	\caption{PriorBased-SR for various $K$ and $\sigma_n$ values. Sparse vector cardinality $\|\alphav\|_0=1$.}
	\label{fig:importance1}
\end{figure}

\begin{figure}[t]
	\centering
	\includegraphics{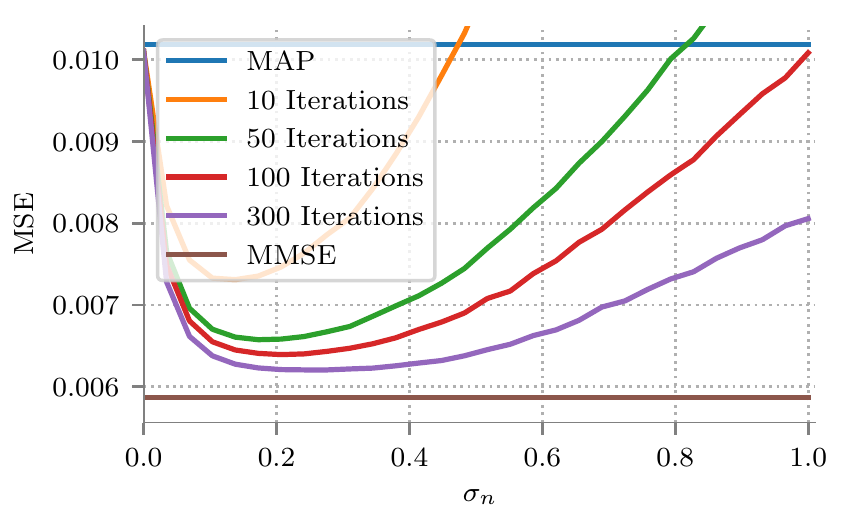}
	\caption{PriorBased-SR for various $K$ and $\sigma_n$ values. Sparse vector cardinality $\|\alphav\|_0=3$.}
	\label{fig:importance2}
\end{figure}


\section{A Practical Variant}
\label{sec:general_algorithm}

\begin{algorithm}[t]
\caption{General SR algorithm}\label{algorithm:SR}
\begin{algorithmic}[1]
\Procedure{General-SR}{$\y, \Dv$, PursuitAlg, $\sigma_n$, $K$}
\For{k=1...K}
\State $\n_k\gets$ SampleNoise($\sigma_n$)
\State $\tilde{\alphav}_k\gets$ PursuitAlg$(\y+\n_k, \Dv)$
\State $\hat{S}_k \gets $ Support$(\tilde{\alphav}_k)$
\If {$P(\alphav|\y,S)$ is known}
\State $\hat{\alphav}_k\gets \hat{\alphav}_{\hat{S}_k}^{\text{Oracle}}(\y)$
\Else
\State $\hat{\alphav}_k\gets \left(\Dv_S^T\Dv_S\right)^{-1}\Dv_S^T\y$
\EndIf
\EndFor
\State $\hat{\alphav}=\frac{1}{K}\sum_{k=1}^{K}\hat{\alphav}_k$
\State \textbf{return} $\hat{\alphav}$
\EndProcedure
\end{algorithmic}
\end{algorithm}

We present the practical variant of Algorithm \ref{algorithm:SR_importance} in Algorithm \ref{algorithm:SR}. Note that we split the algorithm into two cases. In the first we assume that even though we have no knowledge regarding the probability mass function of the support $P(S)$, we do know the probability density function of the non-zero elements given their indices and the measurements, making the oracle estimator obtainable. The second case only assumes knowledge regarding the dictionary $\Dv$, replacing the oracle with a simple Least Squares (LS) operation. 

Computing the MMSE estimator in Equation \eqref{eq:mmseSum} requires a complete knowledge regarding the generative model. Therefore, in cases where the prior is only partially known, the MMSE estimator cannot be obtained, and achieving its MSE performance is generally not feasible. Nevertheless, as we empirically show here and in the following sections, Algorithm \ref{algorithm:SR} succeeds to effectively approximate the MMSE estimator.

Before we go through the analytic arguments and empirical evidence provided in the coming sections, we present some intuition behind the proposed method. Similar to the previous method, the SR noise added will introduce a small perturbation in the signal $\y + \n_k$ in each iteration. Therefore, in each iteration the pursuit will extract supports that are likely to fit the original signal $\x$. The final estimator is an arithmetic mean of all the recovered supports, meaning that the supports that have higher posterior probability will be retrieved more often, making their weight greater than other, less probable supports. If $K$ is large enough and the occurrence of each support resembles its un-normalized posterior probability, then the averaged result coincides with the MMSE. 

We propose the following experiments in order to demonstrate Algorithm \ref{algorithm:SR}'s performance. As in the previous section, we use a normalized random Gaussian dictionary, this time of size $25\times 50$, and generate random sparse vectors with 3 non-zero elements and Gaussian coefficients. We multiply the sparse vectors by the dictionary and add a Gaussian noise to the signals. To obtain clean estimates we use Algorithm \ref{algorithm:SR}, once with BP and once with OMP. Since we assume no knowledge regarding the prior probability of the support of the sparse vector, we use the bounded noise formulation of the pursuit algorithms, i.e.
\begin{align}
\text{(OMP)} \quad &\min_{\alphav} \|\alphav\|_0 \quad \text{s.t.} \quad \|\y - \Dv\alphav\|_2 \leq \epsilon,
\\
\text{(BP)} \quad\quad &\min_{\alphav} \|\alphav\|_1 \quad \text{s.t.} \quad \|\y - \Dv\alphav\|_2 \leq \epsilon.
\end{align}
where $\epsilon$ is chosen to be optimal. We repeat the experiment twice. Once we assume we know the distribution of $\alphav_S|S,\y$ allowing us to use the oracle in Equation \eqref{eq:oracle}, and once using the plain LS variant. We compare the results to those obtained by Algorithm \ref{algorithm:SR_importance} with OMP used as a pursuit, as well as to the MMSE and MAP estimators. 

In a second experiment we further examine the effectiveness of our method compared to standard pursuit algorithms for a varying number non-zero elements. For this experiment, we increased the dimensions of the dictionary to $50 \times 100$ to allow for a large number of non-zero elements in $\alphav$, while keeping it sparse. In all the described experiments we use 300 iterations for both algorithms and average over $10,000$ realizations. The results can be seen in Figure \ref{fig:bp_omp} and \ref{fig:cardinality} respectively.

Examining the results obtained in Figure \ref{fig:cardinality}, our LS variant improves over standard pursuit algorithms for various cardinality values. Furthermore, Figure \ref{fig:bp_omp} demonstrates that while algorithm \ref{algorithm:SR} in its LS and oracle variant improved both BP's and OMP's MSE, the two perform differently. At first, the oracle seems slightly favorable and indeed achieves better results for the optimal $\sigma_n$. Surprisingly however, when too much noise is added their difference diminishes. Comparing Algorithm \ref{algorithm:SR_importance} to Algorithm \ref{algorithm:SR}, the former is much more robust to the standard deviation of the SR noise. The source for this difference is the averaging operator. While Algorithm \ref{algorithm:SR_importance} uses the prior in order to weigh the solutions correctly, in Algorithm \ref{algorithm:SR} each support is weighed by its probability to be chosen in the SR process. In general, the weights given by the SR process do not match the MMSE weights in Equation \eqref{eq:mmseSum}. 

In the following sections we introduce additional assumptions in order to expand the theoretical and empirical analysis presented. First we analyze the case in which the dictionary consists of a unitary matrix and then we analyze the case of a general normalized dictionary, while restricting the cardinality of the sparse vector to be $1$.

\begin{figure}[t]
    \includegraphics{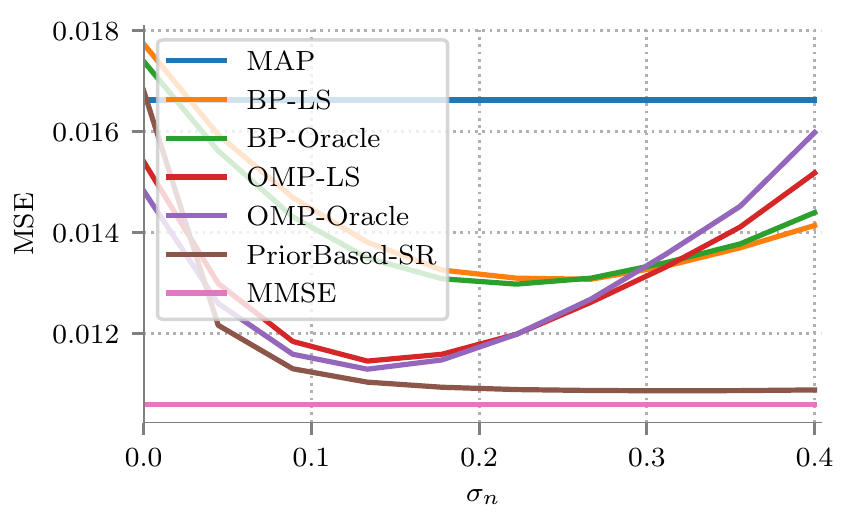}
    \caption{MSE comparison between Algorithm \ref{algorithm:SR_importance} with OMP, and Algorithm \ref{algorithm:SR} with both OMP and BP in its oracle and LS forms.}
    \label{fig:bp_omp}
\end{figure}

\begin{figure}[t]
    \includegraphics{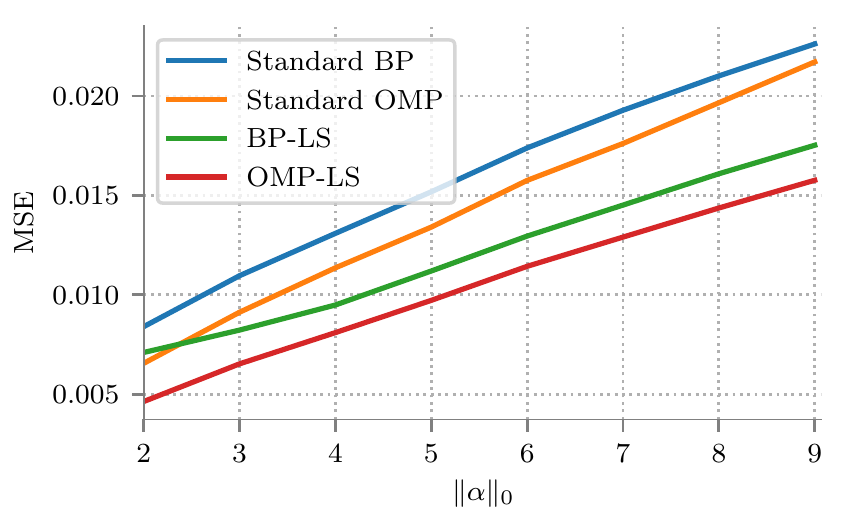}
    \caption{Algorithm \ref{algorithm:SR} with OMP and BP in its LS form compared to standard OMP and BP pursuits vs. the cardinality of the sparse representation vector. $\sigma_n$ is optimal and obtained empirically.}
    \label{fig:cardinality}
\end{figure}


\section{General SR in the Unitary Case}
\label{sec:unitaryPursuit}
\subsection{The Unitary Sparse Estimators}
When the dictionary is a unitary $n\times n$ matrix, we can simplify the expressions associated with the oracle, MAP and MMSE estimators as suggested in \cite{Turek2011, protter2010closed}.
Given a support $S$, the oracle estimator is a constant shrinkage applied on the projected measurements $\betav_S=\Dv_S^T\y$
\begin{equation}
\hat{\alphav}^{\text{Oracle}}_{S}(\y) = c^2\betav_S,
\end{equation}
where $c^{2}=\sigma_{\alpha}^{2}/(\sigma_{\alpha}^{2}+\sigma_{\nu}^{2})$.

The MAP estimator is reduced to the element-wise hard thresholding operator applied on the projected measurements $\betav=\Dv^T\y$, given by
\begin{equation}
\hat{\alpha}_{MAP}\left(\beta\right)=\mathcal{H}_{\lambda_{MAP}}\left(\beta\right)=
\begin{cases}
	c^{2}\beta & \text{if } \left|\beta\right|\geq\lambda_{MAP}, \\
	0 & \text{otherwise}
\end{cases},
\end{equation}
where $\lambda_{MAP} \triangleq \frac{\sqrt{2}\sigma_{\nu}}{c}\sqrt{\log{\left( \frac{1-p_{i}}{p_{i}\sqrt{1-c^{2}}}\right) }}$, and $\alpha$ and $\beta$ are the elements of the vectors $\alphav$ and $\betav$.

The MMSE estimator in the unitary case is a simple elementwise shrinkage operator of the following form:
\begin{equation}
\hat{\alpha}_{MMSE}=\psi(\beta)=\frac{\exp\left({\frac{c^{2}}{2\sigma_{\nu}^{2}}\beta^{2}}\right)\frac{P_{i}}{1-P_{i}}\sqrt{1-c^{2}}}{1 + \exp{\left(\frac{c^{2}}{2\sigma_{\nu}^{2}}\beta^{2}\right)}\frac{P_{i}}{1-P_{i}}\sqrt{1-c^{2}}}c^{2}\beta.
\end{equation}
Note that this shrinkage operator does not result in a sparse vector, just as in the general case.
The above scalar operators are extended to act on vectors in an entry-wise manner.


\subsection{The Unitary SR Estimator}
We now analyze the estimator suggested in Algorithm \ref{algorithm:SR}, under the unitary dictionary assumption.
\begin{proposition}
Let $\Dv$ be a unitary matrix and denote by $Q(\cdot)$ the tail probability of the standard normal distribution\footnote{$Q(\cdot)$ is given by $Q(x)=\frac{1}{\sqrt{2\pi}}\int_{x}^{\infty}e^{-\frac{u^2}{2}} du$.}. Suppose that we use Algorithm \ref{algorithm:SR} with the MAP estimator $\mathcal{H}_{\lambda}$ and white Gaussian SR noise $\n_k\sim\mathcal{N}(\bm{0},\sigma_{n}^2\bm{I})$ as a pursuit. Then, asymptotically, ${\hat{\alpha} = \left[Q\left(\frac{\lambda+\beta}{\sigma_n}\right)+Q\left(\frac{\lambda-\beta}{\sigma_n}\right)\right]}c^{2}\beta$.
\end{proposition}
\begin{proof}
When using the MAP estimator in Algorithm \ref{algorithm:SR}, the thresholding operator is only used to recover the support itself. Once the support is extracted, the final estimator computes the oralce estimator w.r.t. the obtained supports before applying an empirical mean. This process can be equivalently described by the following elementwise \emph{subtractive hard thresholding} operator\footnote{Notice that the written equation operates on the the vectors elementwise.} $\mathcal{H}_{\lambda}^{-}(\cdot)$:
\begin{align}
\hat{\alpha}_{k}\left(\beta,\tilde{n}_{k}\right)=\mathcal{H}^{-}\left(\beta,\tilde{n}_{k}\right) =
\begin{cases}
	c^{2}\beta & \text{if } \left|\beta+\tilde{n}_{k}\right|\geq\lambda_{MAP}, \\
	0 & \text{otherwise}
\end{cases},
\end{align}
where $\tilde{\n}_k = \Dv^T\n_k$. Clearly, since $\Dv$ is unitary, ${\tilde{\n}_k\sim\mathcal{N}\left(\bm{0},\sigma_n^2 \Iv\right)}$. The final estimator is then achieved by an empirical mean,
\begin{equation}
\hat{\alphav}=\frac{1}{K}\sum_{k=1}^{K}\hat{\alphav}_{k}=\frac{1}{K}\sum_{k=1}^{K}\mathcal{H}^{-}_{\lambda}\left(\betav+\tilde{\n}_k\right).
\end{equation}
As $K\to\infty$, the described process asymptotically converges to the expectation
\begin{align}
\mathbb{E}_{\tilde{n}}\left[\mathcal{H}_{\lambda}^{-}\left(\beta,\tilde{n}\right)\right]=&\int_{-\infty}^{\infty}\mathcal{H}_{\lambda}^{-}\left(\beta,\tilde{n}\right)p\left(\tilde{n}\right)d\tilde{n} \\
~
=&\left[Q\left(\frac{\lambda+\beta}{\sigma_n}\right)+Q\left(\frac{\lambda-\beta}{\sigma_n}\right)\right]c^{2}\beta.
\label{eq:qIntegral}
\end{align}
The full derivation can be found in Appendix \ref{appendix:unitary}.
\end{proof}

Unfortunately, analytically bounding the MSE between \eqref{eq:qIntegral} and the MMSE estimator proves to be challenging. However, as we will see shortly, Equation \eqref{eq:qIntegral} numerically approximates the MMSE very accurately.
Note that there are two parameters yet to be set: $\sigma_n$ and $\lambda$. The former tunes the magnitude of the added noise, while the latter controls the value of the thresholding operation. The original MAP threshold $\lambda_{\text{MAP}}$ might be sub-optimal due to the addition of SR noise and therefore, we leave $\lambda$ as a free parameter. We will suggest a method to set these parameters later in this section.


\subsection{Unitary SR Estimation Results}
In order to demonstrate the similarity of the proposed estimator to the MMSE estimator, we compare their shrinkage curves in Figure \ref{fig:shrink}. One can see that, while the curves do not overlap completely, for the right choice of parameters ($\lambda$ and $\sigma_{n}$), the curves are indeed quite close to each other. In terms of MSE, in Figure \ref{fig:unitaryMse} we compare the performance of the general SR method and the MMSE as a function of $\sigma_n$ (with $\lambda$ fixed at the optimal value). Indeed, for the optimal parameters, the two are almost identical. In Appendix \ref{appendix:sure}, the performance of the general SR estimator is demonstrated as a function of both $\lambda$ and $\sigma_n$.

We now discuss how to set the parameters in order to reach these optimal results.

\begin{figure}[t]
    \centering
    \begin{subfigure}{0.5\textwidth}
        \centering
        \includegraphics{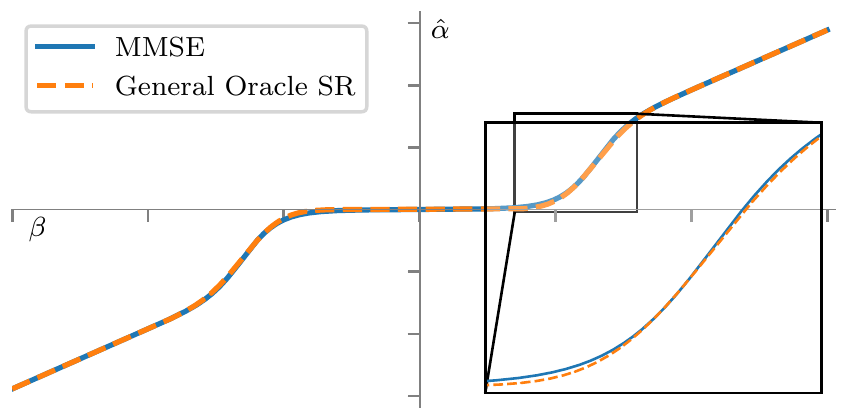}
        \caption{Shrinkage curves comparison.}
        \label{fig:shrink}
    \end{subfigure}
    \begin{subfigure}{0.5\textwidth}
    		\centering
	    	\includegraphics{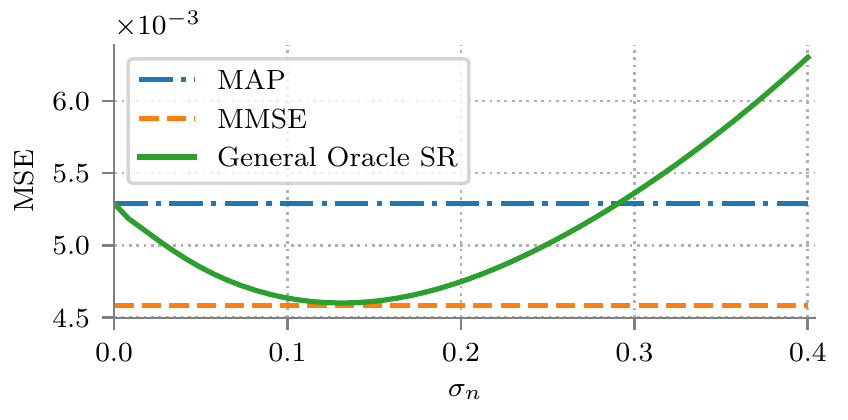}
    		\caption{General SR estimator's MSE for varying $\sigma_n$.}
    		\label{fig:unitaryMse}
    \end{subfigure}
    	\label{fig:unitaryPerformance}
	\caption{The asymptotic general SR estimator vs. the MMSE. $\Dv$ is a unitary $100 \times 100$ dictionary, $\nuv\sim\mathcal{N}\left(\0,\sigma_{\nuv}^{2}\Iv\right)$, $\sigma_{\nuv}=0.2$, $P_i=0.05$ $\forall i\in[m]$ and $\alphav_{s}\sim\mathcal{N}\left(\0,\Iv\right)$.}
\end{figure}


\subsection{Finding the Optimal Parameters for the Unitary Case}
To optimize the free parameters $\lambda$ and $\sigma_n$, we propose to use Stein's Unbiased Risk Estimate (SURE) \cite{stein1981estimation} which measures an estimator's MSE up to a constant when the additive noise is Gaussian. The SURE formulation of the expected MSE is given by
\begin{align}
\mu(\mathcal{H}_{\lambda}^{-}(\betav,\sigma_n), \betav) = \|&\mathcal{H}^{-}_{\lambda}(\betav,\sigma_n)\|_2^2 − 2\mathcal{H}_{\lambda}^{-}(\betav,\sigma_n)^T\betav \\ + &2\sigma_\nu^2 \nabla_{\beta}\mathcal{H}_{\lambda}^{-}(\betav,\sigma_n).
\end{align}
In the unitary case this is further simplified to an element-wise sum:
\begin{align}
\mu(\mathcal{H}_{\lambda}^{-}(\betav,\sigma_n), \betav) &= \sum_i\mu(\mathcal{H}_{\lambda}^{-}(\beta_i,\sigma_n), \beta_i).
\end{align}
Expanding the estimator, one gets
\begin{align}
\mu(\mathcal{H}_{\lambda}^{-}(\betav,\sigma_n), \betav) &=\sum_i\mathcal{H}_{\lambda}^{-}(\beta_i,\sigma_n)^2 − 2\mathcal{H}_{\lambda}^{-}(\beta_i,\sigma_n)\beta_i \\
&\quad\quad+ 2\sigma_\nu^2 \frac{d}{d\beta_i}\mathcal{H}_{\lambda}^{-}(\beta_i,\sigma_n),
\label{eq:sureUnitary}
\end{align}
and we wish to optimize for $\sigma_n$ and $\lambda$:
\begin{equation}
\sigma_n,\lambda = \argmin_{\sigma_n,\lambda}~\mu(\mathcal{H}_{\lambda}^{-}(\betav,\sigma_n), \betav).
\end{equation}
This can be further simplified, as shown in Appendix \ref{appendix:sure}.
Also, this Appendix depicts the surface $\mathbb{E}_{\n}\mu$ for a specific experiment. Interestingly, we observe that the empirically obtained optimal $\lambda$ is quite close to the threshold suggested by the MAP estimator.

Note that in this process, we obtain an MMSE approximation without explicitly knowing the prior distribution of each element $P_i$. Furthermore, if $\sigma_{\alpha}$ is not known, one can easily estimate it as follows: The dictionary is unitary and therefore the mean energy of the signal $\y$ is $\sigma_{\nu}^2 + \sigma_{\alpha}^2$. Assuming we know $\sigma_{\nu}$ one can easily obtain $\sigma_{\alpha}$.


\section{General SR in the Single Atom Case}
\label{sec:singleAtom}

\subsection{Cardinality $1$ Performance}
While the unitary case is simpler to analyze, most applications rely on overcomplete dictionaries. In Section \ref{sec:general_algorithm}, we already showed that empirically, Algorithm \ref{algorithm:SR} improves the MSE performance of standard pursuit algorithms in the general case. We now try to further analyze Algorithm \ref{algorithm:SR} by introducing a single atom assumption, i.e. assuming that the cardinality of the sparse vectors is restricted to one. From \cite{Elad2009} we have that in this case, the MAP estimator described in Section \ref{sec:bayes} boils down to the following form:
\begin{equation}
\hat{\alpha}_i^{\text{MAP}}(\y)=\begin{cases}c^2\beta_{\hat{S}},&i=\hat{S}\\0,&i\neq \hat{S}\end{cases},\quad \beta_{\hat{S}}=\dv_{\hat{S}}^T\y,\quad c^2=\frac{\sigma_{\alpha}^2}{\sigma_{\alpha}^2 + \sigma_{\nu}^2},
\end{equation}
where $\hat{\alpha}_i^{\text{MAP}}$ is the $i$-th index in the vector $\hat{\alphav}^{\text{MAP}}$, $\dv_i$ is the $i$th atom and $\hat{S}$ represents the chosen atom index:
\begin{align}
\hat{S}&=\argmin_{S} \|\beta_S\dv_S-\y\|^2_2=\argmax_{S}{\left|\dv_S^T\y\right|}.
\end{align}

\begin{proposition}
Let $\Dv$ be a dictionary with normalized atoms and $\alphav$ a sparse representation vector such that $\|\alphav\|_0=1$, and suppose we use Algorithm \ref{algorithm:SR} with the MAP estimator as a pursuit algorithm. Then, asymptotically, one obtains $\hat{\alphav}$ such that its $i$-th index is $\hat{\alpha}_i= c^2\dv_i^T\y\cdot P(\hat{S}=i)$, where $P(\hat{S}=i)$ is the probability of the SR process to retrieve the support $\hat{S}$.
\label{proposition:singleAsymptotic}
\end{proposition}

\begin{proof}
Following the subtractive hard thresholding concept suggested in the previous section, we introduce the following SR estimator:
\begin{equation}
\hat{\alpha}_{k_i}(\y,\n_k)=\begin{cases}c^2\beta_{\hat{S}_{\text{SR}}},&i=\hat{S}_{\text{SR}}\\0,&i\neq \hat{S}_{\text{SR}}\end{cases},\quad \beta_{\hat{S}_{\text{SR}}}=\dv_{\hat{S}_{\text{SR}}}^T\y,
\end{equation}
where this time the chosen index $\hat{S}_{\text{SR}}$ is affected by an additive SR noise:
\begin{align}
\hat{S}_{\text{SR}}=&\argmin_{S} \|\dv_S^T(\y+\n)-(\y+\n)\|^2_2\\=&\argmax_{S} \left|\dv_S^T\left(\y + \n_k\right)\right|.
\label{eq:maxAtom}
\end{align}
Hence, asymptotically, Algorithm \ref{algorithm:SR} converges to
\begin{align}
\mathbb{E}_{\n}\left[\hat{\alphav}(\y,\n)\right]=&\mathbb{E}_S\left[\mathbb{E}_{\n|S}\left[\hat{\alphav}(\y,\n)\middle|S\right]\right]
\\
=&\sum_{i=1}^m \mathbb{E}_{\n|S}\left[\hat{\alphav}\middle| S=i\right] P\left(\hat{S}=i\right)\\
=&c^2\begin{bmatrix} \beta_1 P\left(\hat{S}=1\right) \\ \vdots \\ \beta_m P\left(\hat{S}=m\right) \end{bmatrix}.
\end{align}
as claimed. 
\end{proof}

As before, the difference between the general MMSE estimator \eqref{eq:mmseSum} and Algorithm in Equation \ref{algorithm:SR} in its oracle form is in the weight assigned to each solution $P(\hat{S})$. We now further analyze the weight obtained in the SR process under the single atom assumption. As stated in \eqref{eq:maxAtom}, the chosen atom $i$ is the most correlated one with the input SR noisy signal:
\begin{align}
P(\hat{S}=i)&=P\left(\left|\dv_i^T(\y+\n)\right| > \max_{j\neq i} \left| \dv_j^T(\y+\n)\right|\right)\\
&=P\left( \left| \tilde{n}_i \right| > \max_{j \neq i} \left| \tilde{n}_j \right| \right),
\label{eq:atomProbability}
\end{align}
where we defined $\tilde{\n}\triangleq \Dv^T\left(\y+\n\right) $. Choosing $\n\sim\mathcal{N}\left(\bm{0},\sigma_n^2\bm{I}\right)$ then $\tilde{\n}$ is Gaussian as well:
\begin{equation}
\tilde{\n}=\begin{bmatrix}
\tilde{n}_1\\
\vdots\\
\tilde{n}_m
\end{bmatrix}
\sim \mathcal{N} \left( \Dv^T\y, \sigma_n^2 \Dv^T \Dv\right).
\label{eq:projNoise}
\end{equation}
Therefore, the probability of choosing the $i$-th atom is distributed as the probability of the maximum value of a random Gaussian vector with \textit{correlated} variables, since in the non-unitary case, $\Dv^T\Dv$ is not a diagonal matrix.
Facing this difficulty, we propose to tackle it as follows:
\begin{enumerate}
\item Instead of adding the SR noise to $\y$, we can add it to the projected signal $\Dv^T\y$, thus avoiding the variables $\left\{\tilde{n}_i\right\}_{i=1}^m$ being correlated.

\item Add statistical assumptions regarding the dictionary $\Dv$, leading to average case conclusions.

\item Change the pursuit used. Intuitively, using the MAP will produce the optimal results, since it retrieves the most probable support for any given signal. However, changing the pursuit might ease the analysis of the asymptotic estimator. We leave the study of this option for future work.
\end{enumerate} 
We now analyze the first two proposed alternatives.


\subsection{Add Noise to the Representation Domain}
\begin{proposition}
Let the same conditions as Proposition \ref{proposition:singleAsymptotic} hold. Moreover, let $\n\sim\mathcal{N}\left(\bm{0},\sigma_n^2\bm{I}\right)$ be an SR noise added to the representation domain. Then, the probability to retrieve the $i$-th support is given by
\begin{multline}
P(\hat{S}=i)=\int_0^\infty \frac{1}{\sqrt{2\pi}\sigma_n}\left[e^{-\frac{\left(t+\beta_i\right)^2}{2\sigma_n^2}} + e^{-\frac{\left(t-\beta_i\right)^2}{2\sigma_n^2}}\right]
\\
\times \prod_{j \neq i}\left[1-\left(Q\left( \frac{t-\beta_j}{\sigma_n} \right) + Q\left( \frac{t+\beta_j}{\sigma_n} \right)\right)\right]dt.
\label{eq:oneAtomProb}
\end{multline}
\label{proposition:projNoise}
\end{proposition}

\begin{proof}
We continue from \eqref{eq:atomProbability}, only now the noise $\tilde{n}_i$ is white and has the following properties:
\begin{equation}
\tilde{\n}
\sim \mathcal{N} \left( \Dv^T\y, \sigma_n^2 \boldsymbol{I}_{m \times m}\right).
\end{equation}
Plugging this into \eqref{eq:atomProbability}: 
\begin{align}
P(\hat{S}=i)&=P\left(\left|\dv_i^T\y+n_i\right| > \max_{j\neq i} \left| \dv_j^T\y+n_j\right|\right)
\\
&=P\left( \left| \tilde{n}_i \right| > \max_{j \neq i} \left| \tilde{n}_j \right| \right)
\\
&=\int_0^\infty P\left( \max_{j \neq i} \left| \tilde{n}_j \right| < t \middle| \left| \tilde{n}_i \right|=t \right) P \left( \left| \tilde{n}_i \right| = t \right)dt
\end{align}
\begin{align}
&=\int_0^\infty P\left( \max_{j \neq i} \left| \tilde{n}_j \right| < t \right) P \left( \left| \tilde{n}_i \right| = t \right)dt.
\label{eq:supportProbability}
\end{align}
For the first term, the elements of $\tilde{\n}$ are independent, and therefore
\begin{align}
P\bigg( \max_{j \neq i} \left| \tilde{n}_j \right| & < t \bigg) = \prod_{j\neq i}P\left(\left| \tilde{n}_j \right| < t \right) 
\\
&= \prod_{j\neq i}\left[ 1 - P\left(\left| \tilde{n}_j \right| > t \right)\right] 
\\
&=\prod_{j\neq i}\left[ 1 - \left(Q\left(\frac{t+\beta_j}{\sigma_n} \right) + Q\left(\frac{t-\beta_j}{\sigma_n} \right) \right)\right],
\end{align}
where the last equality follows similar steps as in Appendix \ref{appendix:unitary}. The second term in \eqref{eq:supportProbability} is simply the PDF of the absolute value of a Gaussian variable, therefore
\begin{align}
P\left( \left| \tilde{n}_i \right| = t  \right)&=\frac{1}{\sqrt{2\pi\sigma_n^2}}\left( e^{-\frac{\left(t-\beta_i\right)^2}{2\sigma_n^2}} + e^{-\frac{\left(t+\beta_i\right)^2}{2\sigma_n^2}}\right).
\end{align}
Putting the two terms back into \eqref{eq:supportProbability} we obtained the claimed relation in Equation (\ref{eq:oneAtomProb}).
\end{proof}

The obtained expression cannot be solved analytically but can be computed numerically. We now empirically examined the properties of the derived estimator. We generated a random dictionary of size $25\times 50$ with iid Gaussian elements and normalized the atoms. Then, we generated sparse vectors with a single non-zero element with a Gaussian value $\alpha_S\sim\mathcal{N}\left(0,1\right)$. Noisy measurements were generated by multiplying the dictionary with the sparse vectors and adding noise ${\nuv\sim\mathcal{N}\left(\bm{0},\bm{I}\sigma_{\nu}^2\right)},~ \sigma_{\nu}=0.2$. In Figure \ref{fig:oneAtomMSE} we compare the MSE of the MMSE and MAP estimators to Algorithm \ref{algorithm:SR} when the noise is added to the representation domain. Indeed, the proposed method improves the MSE of the MAP estimator and almost achieves the MMSE estimator's performance for the right choice of $\sigma_n$.

In Figure \ref{fig:oneAtomProb} we show the probability of recovering the true support $P_{\text{success}}$ as a function of $\sigma_n$, both from \eqref{eq:oneAtomProb} and from iterating Algorithm \ref{algorithm:SR} 100 times, each time picking the most correlated atom. We also compare it to the MMSE weight from \eqref{eq:mmseGeneral}. The optimal $\sigma_n$ in terms of MSE is drawn as a vertical dashed line. Notice that for the optimal choice of $\sigma_n$, the probability of the true support to be chosen is similar to that given in the MMSE solution. In other words, the optimal $\sigma_n$ is the one that approximates the weight of the support to the weight given by the MMSE expression. The trend in \eqref{eq:oneAtomProb} shows, unsurprisingly, that as we add noise, the probability of successfully recovering the true support decreases. In the limit, when $\sigma_n \rightarrow \infty$ the signal will be dominated by the noise and the success probability will be uniform among all the atoms, i.e. equal to $P_{success} = \frac{1}{m}$.
\begin{figure}[t]
    \includegraphics{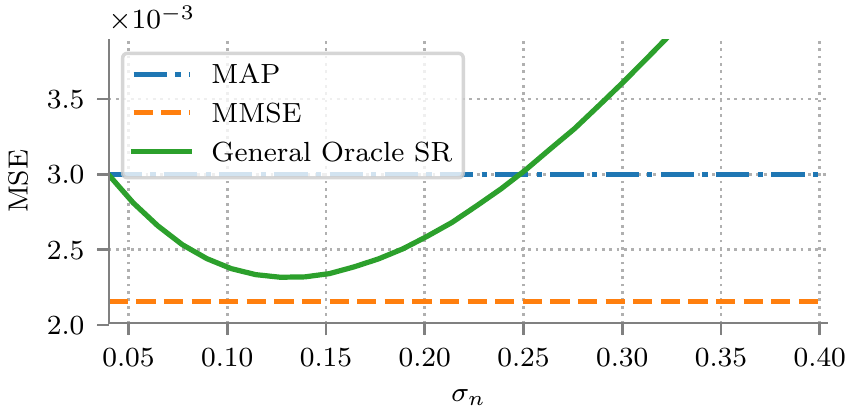}
    \caption{MSE comparison of the MAP, MMSE and 100 iterations of General Oracle SR with non unitary overcomplete dictionary and a sparse representation with 1 non-zero element.}
    \label{fig:oneAtomMSE}
\end{figure}
\begin{figure}[t]
	\includegraphics{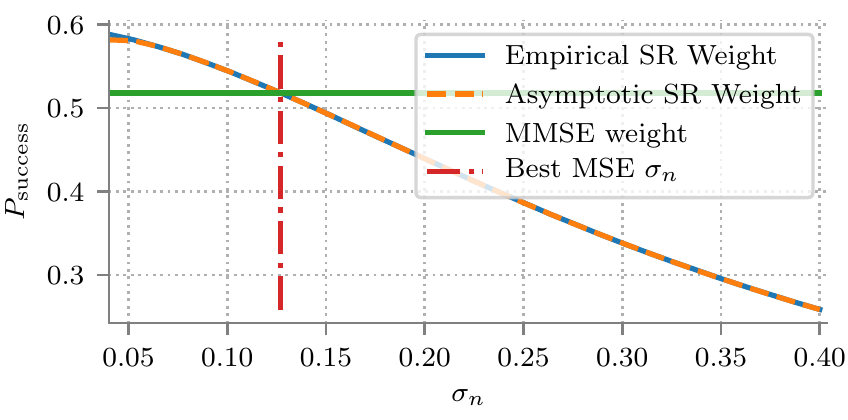}
    \caption{Numerical integration of $P(S=\text{True Suport}|y)$ and the empirical weight achieved by 100 iterations of SR.}
    \label{fig:oneAtomProb}
\end{figure}

Due to these findings, and since the optimal MSE is comparable to that of the MMSE estimator, one might expect a similar behavior for most of the other possible supports. To examine this idea, we carried out the following experiment. We randomized many representations $\alphav$, all containing a non-zero coefficient in the same index. Then, we ploted the histogram of the average empirical probability of each element in the vector $\alphav$ to be non-zero (obtained by pursuit). Finally, we compared these probabilities to the weights of the MMSE from \eqref{eq:mmseGeneral}. This experiment was repeated for various $\sigma_n$ values and for each such value we compared the entire support histogram. We expect the two histograms (SR and MMSE) to fit for the right choice of added noise parameter $\sigma_n$. In Figure \ref{fig:oneAtomHistograms} we see the results of the described experiment. 

Analyzing the results obtained, we see that when no noise is added (this is the average case of the MAP estimator), apart from the true support, the elements have a much lower weight than the MMSE. As noise is added, the true support's probability decreases and its weight is divided among the other elements. At some point the two histograms almost match each other completely. At that point, the SR MSE almost equals that of the MMSE. As we add more noise, the true support's probability keeps decreasing and the other elements keep increasing and the histograms are now farther apart from each other. When we reach $\sigma_n \rightarrow \infty$ we obtain uniform probability for all the supports.

To further demonstrate their similarity, the left axis in Figure \ref{fig:oneAtomHistogramsDist} is the $D_{K\|L}$ distance (Kullback-Leibler divergence) between the two histograms, and the right axis is the MSE. As expected, when their $K\|L$ divergence is the smallest, the MSE is minimal.
\begin{figure*}[t]
	\includegraphics{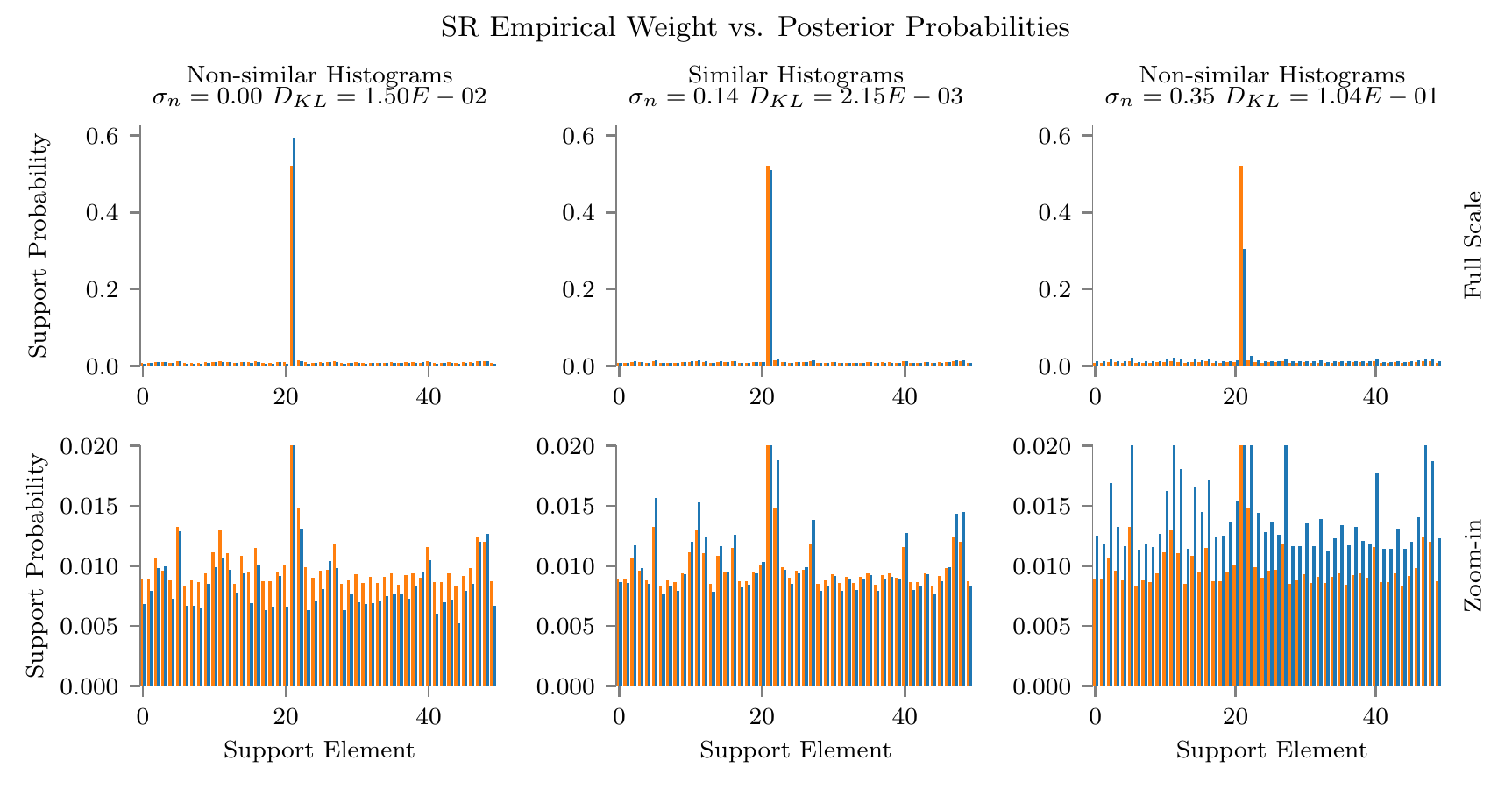}
    \caption{{\color{orange}MMSE (in orange)} and 100 iterations of {\color{RoyalBlue}General-SR (in blue)} support weights histograms for varying values of $\sigma_n$; Top row presents the entire scale; Bottom row is zoomed-in to emphasize the differences in the smaller weights; Title for each column shows the $\sigma_n$ value and KL divergence between the histograms; Atom number 21 is the true support.}
    \label{fig:oneAtomHistograms}
\end{figure*}
\begin{figure}[t]
    \centering
    \includegraphics{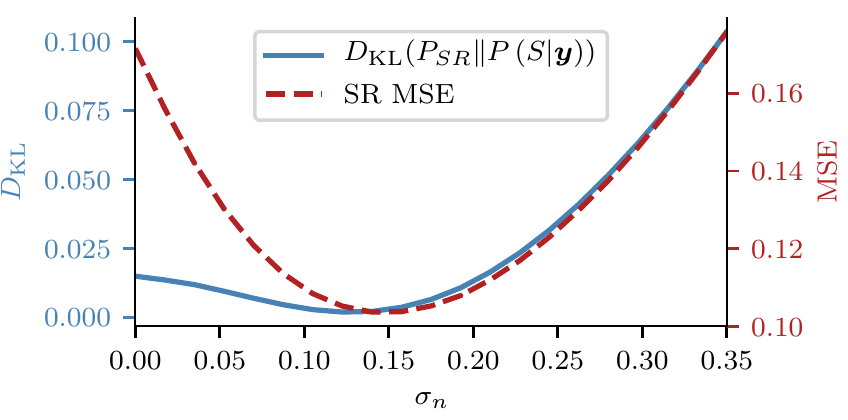}
    \caption{Subtractive SR MSE and $D_{K\|L}$ divergence between the MMSE and SR weights. When the divergence is small so is the MSE.}
    \label{fig:oneAtomHistogramsDist}
\end{figure}

\subsection{Statistical Assumptions on the Dictionary $\Dv$}
\label{sec:dictionaryPrior}
In this section we will try to simplify the expression in \eqref{eq:atomProbability} by assuming that the columns of the dictionary are statistically uncorrelated.
Formally, our assumption is that the atoms $\dv_i$ are drawn from some random distribution that obeys the following properties:
\begin{equation}
\mathbb{E} \left[\dv_i^T \dv_j\right]=0, \quad i \neq j \quad 1\leq i,j \leq m,
\label{eq:dictAssumption1}
\end{equation}
and that the atoms are normalized:
\begin{equation}
\|\dv_i\|_2=1, \quad 1\leq i \leq m.
\label{eq:dictAssumption2}
\end{equation}
\begin{proposition}
Let the same conditions as Proposition \ref{proposition:singleAsymptotic} hold and furthermore suppose that the dictionary's atoms are statistically uncorrelated. Then, when using Algorithm \ref{algorithm:SR} with the MAP estimator, adding white Gaussian SR noise to the signal domain with standard deviation $\sigma_n$ is equivalent to adding white Gaussian SR noise to the representation domain with the same standard deviation $\sigma_n$.
\label{proposition:uncorrelated}
\end{proposition}
\begin{proof}
From \eqref{eq:projNoise}, we have $\tilde{\n}=\Dv^T(\y+\n)$.
Observe that given the dictionary $\Dv$, each of the elements in this vector is a Gaussian random variable:
\begin{multline}
\tilde{n}_i|\dv_i = \dv_i^T(\y+\n) = \sum_{k=1}^{n}d_{i,k}(y_k+n_k)=
\\
\sum_{k=1}^{n}d_{i,k}y_k + \sum_{k=1}^{n}d_{i,k}n_k=\mu_i+\sum_{k=1}^{n}d_{i,k}n_k.
\end{multline}
Given the measurements and the dictionary, the first sum $\sum_{k=1}^{n}d_{i,k}y_k\triangleq \mu_i$ is some constant. The second term in the expression is a weighted sum of $n$ iid Gaussian random variables $\{n_k\}_{k=1}^n$, hence it is Gaussian. Clearly, its mean value is 0, and its standard deviation is $\sigma_n$, hence $\tilde{n}_i|\dv_i \sim \mathcal{N}(\dv_i^T\y, \sigma_n)$ for $n\rightarrow\infty$.

Now we turn to analyze the properties of the entire vector $\tilde{\n}$. From the previous analysis, given the dictionary $\Dv$, $\tilde{\n}$ is a random Gaussian vector with the mean vector $\muv_{\tilde{\n}}|\Dv=\Dv^T\y$. Using the properties of the noise $\mathbb{E}\left[\n\n^T\right]=\sigma_n^2I$, the auto-correlation matrix of $\tilde{\n} | \Dv$ is by definition:
\begin{equation}
\Sigmav | \Dv =\mathbb{E} \left[\Dv^T\n\n^T\Dv\middle|\Dv\right]=\Dv^T \mathbb{E}\left[\n\n^T\right] \Dv=\sigma_n^2\Dv^T\Dv.
\end{equation}
Analyzing the average case, the mean vector is of the form:
\begin{equation}
\muv_{\tilde{\n}}=\mathbb{E}_{\Dv}\left[\Dv^T\y\right],
\end{equation}
and the auto-correlation matrix is simply diagonal:
\begin{align}
\Sigmav=&\mathbb{E_{\Dv}}[\Sigmav | \Dv]= \mathbb{E}\left[\sigma_n^2\Dv^T\Dv\right]=\sigma_n^2\Iv,
\end{align}
where we used the assumptions in \eqref{eq:dictAssumption1} and \eqref{eq:dictAssumption2}.
\end{proof}

From Proposition \ref{proposition:uncorrelated}, the uncorrelated atoms assumption leads $\tilde{\n}$ to have the same properties as in Proposition \ref{proposition:projNoise}. Therefore, the empirical analysis following Proposition \ref{proposition:projNoise} holds for this case as well.

To demonstrate empirically that indeed the two are the same, we performed the following experiment. We sampled a random dictionary $\Dv$ and random sparse representations $\alphav$ with cardinality of $1$ as the generative model described earlier suggests. In this experiment we used a dictionary $\Dv$ of size $200 \times 400$ and $2000$ random sparse representations. Using the generated vectors and dictionary we created signals $\y$: $\y=\Dv\alphav + \nuv$. To denoise the signals, we once used Algorithm \ref{algorithm:SR} with noise $\n_{\text{sig.}} \sim \mathcal{N}(\0, \sigma_n^2 \Iv_{n \times n})$ added to the signal vectors $\y + \n_{\text{sig.}}$, and once with noise $\n_{\text{rep.}} \sim \mathcal{N}(\0, \sigma_n^2 \Iv_{m \times m})$ added to the representation domain $\Dv^T\y + \n_{\text{rep.}}$. As before, we use the MAP estimator as the chosen pursuit. In Figure \ref{fig:noiseLocation} we see that the MSE of the two cases result in an almost identical curve. Small differences exist due to the finite dimensions used in the experiment.
\begin{figure}[t]
	\centering
	\includegraphics{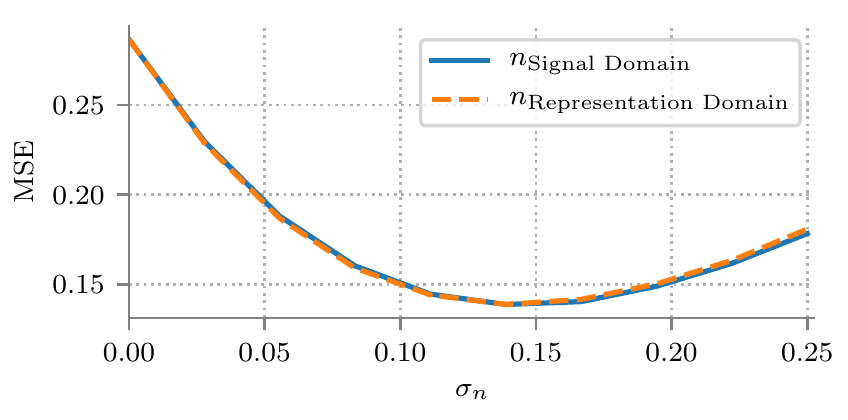}
	\caption{Noise location comparison. 500 iterations of General Oracle SR with MAP estimator as a pursuit method. $\nuv\sim\mathcal{N}\left(\0,\sigma_{\nuv}^{2}I\right)$, $\sigma_{\nuv}=0.2$, $\left|\left|\alphav\right|\right|_{0}=1$ and $\alphav_{s}\sim\mathcal{N}\left(0,1\right)$. The SR noises are $\n_{\text{Signal Domain}} \sim \mathcal{N}(0, \sigma_n I_{n \times n})$. ${\n_{\text{Representation Domain}} \sim \mathcal{N}(0, \sigma_n I_{m \times m})}$.}
	\label{fig:noiseLocation}
\end{figure}

Note that the noise energy added in the representation domain is much larger than that of the noise added to the signal, i.e. $E\|\n_{\text{rep.}}\|_2^2=m\sigma_n^2 > n\sigma_n^2 = E\|\n_{\text{sig.}}\|_2^2$ but the results remain the same due to the unit norm of the dictionary $\|\dv_i\|_2=1$, and of course, the uncorrelated assumption.

To conclude this subsection, statistically uncorrelated atoms provide a way to further our theoretical analysis. In this case, adding noise in the signal domain converges to the analysis addressed in the previous subsection, where noise is instead added in the representation domain. As a result, similar results and conclusions can be drawn.


\section{What Noise Should be Used?}
\label{sec:noise}
Throughout this work we naturally used white Gaussian noise by default. In this section we question this decision and wonder whether we can use noise models with different distributions and whether it affects the performance of the stochastic resonance estimator.

\begin{theorem}
Let $\n$ be a random vector with iid elements sampled from a distribution with finite mean and variance and used as SR noise in Algorithm \ref{algorithm:SR}. Then, as the dimension of the sparse representation grows asymptotically, the estimate given by Algorithm \ref{algorithm:SR} is not affected by $\n$'s distribution.
\end{theorem}
\begin{proof}
Denoting $\tilde{\n} \triangleq \Dv^T\n$, each element $\tilde{n}_i$ is:
\begin{equation}
\tilde{n}_i = \dv_i^T\n = \sum_{j=1}^md_{i,j}n_j.
\end{equation}
Without loss of generality, we assume that the atoms are normalized, i.e., $\|\dv_i\|_2=1$. The above expression is a weighted average of the variables $\{n_j\}_{j=1}^{m}$. Since $\{n_j\}_{j=1}^m$ are iid and have bounded mean and variance, then the central limit theorem holds. Therefore, as $m$ increases, $\tilde{n}_i$ is asymptotically Gaussian regardless of the distribution of the original additive noise $n_j$.
\end{proof}
Following the previous statement, we experimented with a different distribution for a random noise vector. We employed an element-wise iid uniform noise with $0$ mean $n_{\mathcal{U}} \sim \mathcal{U}[-r, r]$. In order to be compatible with a Gaussian noise $n_{\mathcal{N}} \sim \mathcal{N}(0,\sigma_n^2)$ we chose $r=\sqrt{3}\sigma_n$ thus assuring the same standard deviation for the two cases. In Figure \ref{fig:uniform} we compare the random Gaussian noise with the uniform one as described, and indeed, the curves overlap.

In Appendix \ref{appendix:benoulliNoise} we further experiment with a different form of SR noise, leading to similar performance in terms of MSE, while reducing the computations performed.

\begin{figure}[t]
	\includegraphics{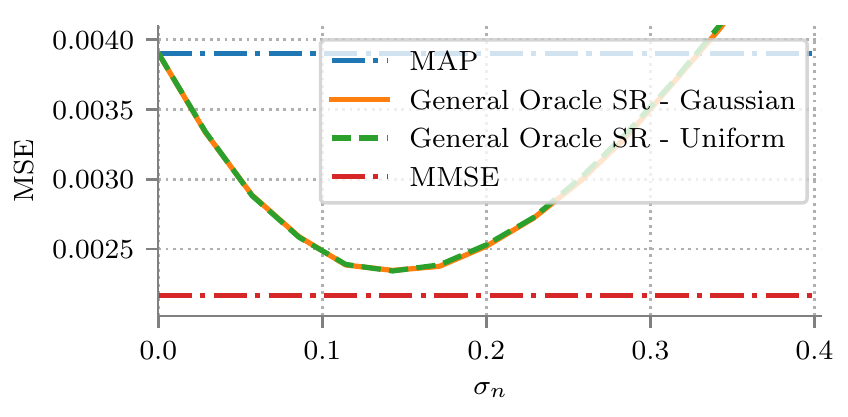}
	\caption{Uniform vs. Gaussian SR noise. $\Dv \in \mathbb{R}^{50 \times 100}$, ${\nuv\sim\mathcal{N}\left(\0,\sigma_{\nuv}^{2}I\right)}$, $\sigma_{\nuv}=0.2$, $\left|\left|\alphav\right|\right|_{0}=1$ and $\alphav_{s}\sim\mathcal{N}\left(0,1\right)$. 100 iterations of Algorithm \ref{algorithm:SR} with the MAP estimator as a pursuit. }
	\label{fig:uniform}
\end{figure}


\section{Image Denoising}
\label{sec:image}
In this section we demonstrate the benefits of using Algorithm \ref{algorithm:SR} in image denoising. We use the Trainlets \cite{sulam2016trainlets} dictionary trained on facial images from the Chinese Passport dataset as described in \cite{sulam2016large}. In the dataset, each image is of size $100 \times 100$ pixels and contains a gray-scale aligned face. The application we address is denoising, that is approximating the following optimization problem:
\begin{equation}
\hat{\alphav} = \argmin_{\alphav} \|\Dv\alphav - \y\|_2 \quad \text{s.t.} \quad \|\alphav\|_0=L.
\end{equation}
Since this problem is intractable, an approximation is achieved using the Subspace Pursuit (SP) algorithm \cite{dai2009subspace}, which provides a fast converging algorithm for a fixed number of non-zeros $L$. The particular choice of using SP follows from the prohibitive cardinality and dimensions for, rendering the OMP as a highly non-efficient alternative.

In our experiment, we corrupt an unseen image from the dataset with additive white Gaussian noise, using various standard deviation $\sigma_{\nu}$ values. Then, we denoise the image using SP, where the number of non-zeros $L$ is empirically set to maximize the denoising performance. We then apply Algorithm \ref{algorithm:SR} in its LS form, using 200 iterations and the same SP settings. Note that we do not seek state of the art denoising results but rather to show that our method can be easily applied to improve real image processing tasks.

In Figure \ref{fig:srImage} the results for $\sigma_{\nu}=40$ can be seen. Importantly, the SR result yield a clearer image with much less artifacts. Figure \ref{fig:srSigman} presents the effectiveness of SR under varying SR noise $\sigma_n$. We see that a gain of almost 2 dB is achieved by using SR with a proper $\sigma_n$ over the regular pursuit. Figure \ref{fig:srSigmav} presents a comparison of the PSNR values obtained by SP and Algorithm \ref{algorithm:SR} with SP for varying values of standard deviation $\sigma_{\nu}$. In all of the described experiments, Algorithm \ref{algorithm:SR} improved the denoising results. Generally we observe that as the noise increases, the improvement becomes more significant.

\begin{figure}[t]
	\centering
	\begin{subfigure}[t]{0.2\textwidth}
	\includegraphics[width=1\textwidth,trim={270 370 270 347},clip]{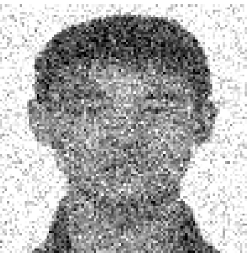}
	\caption{\centering{Noisy image. PSNR=16.1 dB.}}
	\end{subfigure}
	\begin{subfigure}[t]{0.2\textwidth}
	\includegraphics[width=1\textwidth,trim={270 370 270 347},clip]{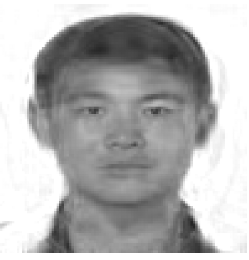}
	\caption{\centering{Subspace Pursuit. PSNR=26.88 dB.}}
	\end{subfigure}\\
	\begin{subfigure}[t]{0.2\textwidth}
	\includegraphics[width=1\textwidth,trim={270 370 270 347},clip]{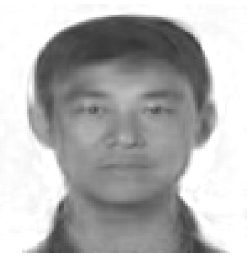}
	\caption{\centering{Stochastic Resonance. PSNR=28.76 dB.}}
	\end{subfigure}
	\begin{subfigure}[t]{0.2\textwidth}
	\includegraphics[width=1\textwidth,trim={270 370 270 347},clip]{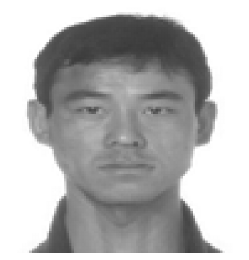}
	\caption{\centering{Clean Image}}
	\end{subfigure}
	\caption{Denoising results comparison. $\sigma_{\nu}=40$, $L=90$.}
	\label{fig:srImage}
\end{figure}

\begin{figure}
	\includegraphics{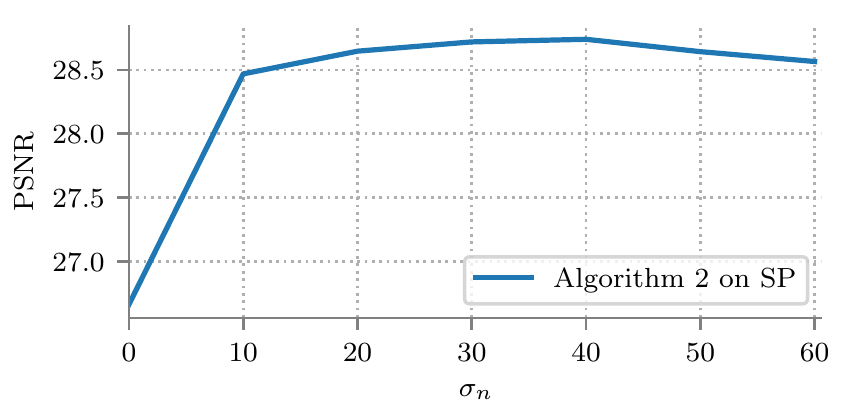}
	\caption{SR results with varying $\sigma_n$ for a noisy image with $\sigma_{\nu}=40$, PSNR=16.1 dB. $\sigma_n=0$ effectively does not use SR.}
	\label{fig:srSigman}
\end{figure}

\begin{figure}[t]
	\includegraphics{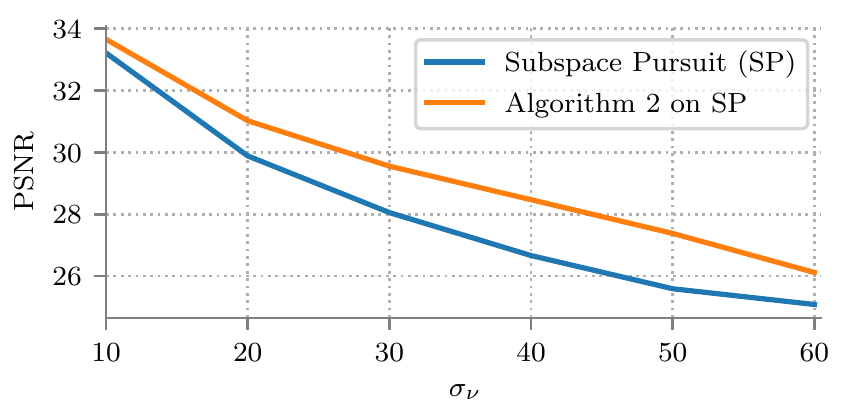}
	\caption{SP and SP + Algorithm \ref{algorithm:SR}: Comparison for varying standard deviation values $\sigma_v$.}
	\label{fig:srSigmav}
\end{figure}

\section{Conclusion}
\label{sec:conclusion}
In this work we suggested two algorithms leveraging the idea of stochastic resonance under the context of sparse coding. We analyzed their theoretical properties and showed that they enable to efficiently deploy arbitrary pursuit algorithms while boosting their performance with SR, providing an approximation to the MMSE estimator. While the first method we suggested is provably convergent to the MMSE, it is not directly applicable in cases where the prior of the sparse vectors is not known. This brought us to introduce a relaxed and more practical alternative. We have analyzed the properties of this second path in several cases and demonstrated its superiority over standard pursuit algorithms in both synthetic cases and on a natural image denoising task. In contrast to previous MMSE approximation methods, the ones suggested in this work have the ability to use any pursuit algorithm as a ``black box'', thus opening the door for MMSE approximation in large dimension regimes for the first time.


%

\appendices
\section{Unitary General SR Asymptotic Estimator}\label{appendix:unitary}
Placing the normal distribution function into \eqref{eq:qIntegral}, we obtain:
\begin{align}
\mathbb{E}_{n}&\left[\mathcal{H}_{\lambda}^{-}\left(\beta,n\right)\right]=\int_{-\infty}^{\infty}\mathcal{H}^{-}_{\lambda}\left(\beta+n\right)p\left(n\right)dn \\
=&\int_{\left|\beta+n\right|\geq\lambda}c^{2}\beta p\left(n\right)dn \\
=&c^2\left[\int_{-\infty}^{-\lambda-\beta}\beta p\left(n\right)dn + \int_{\lambda-\beta}^{\infty}\beta p\left(n\right)dn\right]
\\
=&c^2\beta\left[\int_{-\infty}^{-\lambda-\beta}\frac{1}{\sqrt{2\pi\sigma_n^2}}e^{-\frac{n^2}{2\sigma_n^2}}dn + \int_{\lambda-\beta}^{\infty}\frac{1}{\sqrt{2\pi\sigma_n^2}}e^{-\frac{n^2}{2\sigma_n^2}}\right]
\\
=&c^{2}\beta\left[Q\left(\frac{\lambda+\beta}{\sigma_n}\right)+Q\left(\frac{\lambda-\beta}{\sigma_n}\right)\right].
\end{align}
\section{The Sure Surface for the Unitary Case}
\label{appendix:sure}
Plugging the subtractive hard thresholding $\mathcal{H}_{\lambda}^-$ into \eqref{eq:sureUnitary} leads to the following expression:
\begin{align}
\mathbb{E}_n&\left[\mu\left(\mathcal{H}_{\lambda}^-\right)\right]=\sum_i\left(c^{2}\beta_i\left[Q\left(\frac{\lambda+\beta_i}{\sigma_n}\right)+Q\left(\frac{\lambda-\beta_i}{\sigma_n}\right)\right]\right)^2
\\
&-\sum_i2c^2\beta_i^2\left[Q\left(\frac{\lambda+\beta_i}{\sigma_n}\right)+Q\left(\frac{\lambda-\beta_i}{\sigma_n}\right)\right]
\\
&+\sum_i2\sigma_\nu^2c^2\left[Q\left(\frac{\lambda+\beta_i}{\sigma_n}\right)+Q\left(\frac{\lambda-\beta_i}{\sigma_n}\right)\right]
\\
&+\sum_i2\sigma_\nu^2c^2\beta_i\left[\frac{1}{\sqrt{2\pi}\sigma_n}e^{-\frac{\left(\lambda-\beta_i\right)^2}{2\sigma_n^2}} - \frac{1}{\sqrt{2\pi}\sigma_n}e^{-\frac{\left(\lambda+\beta_i\right)^2}{2\sigma_n^2}}\right].
\end{align}
In order to show that it is indeed easy to optimize $\lambda$ and $\sigma$ on the SURE surface, we demonstrate it by the following experiment. We generate sparse vectors with probability of $P_i=0.01$ for any coefficient to be non-zero. The coefficients of the non-zero entries are drawn from a Gaussian distribution $\mathcal{N}\left(0,1\right)$. We then generate signals using a unitary dictionary and add random Gaussian noise $\mathcal{N}\left(0,0.2^2\right)$. We then compute $\mu\left(\mathcal{H}_{\lambda}^-\right)$ for various $\lambda$ and $\sigma_n$ values. Figure \ref{fig:sure} presents the surface for these values, and Figure \ref{fig:sure_mse} shows the MSE results respectively. We can see that the SURE surface behaves just like the true MSE up to an additive constant and that it is smooth and rather easy to optimize. In terms of MSE, we see the superiority of the proposed estimator over the MAP estimator, and that it is quite close to the MMSE.
\begin{figure*}[t]
    \begin{subfigure}[t]{0.5\textwidth}
        \includegraphics[width=1\textwidth,trim={0 0 0 20},clip]{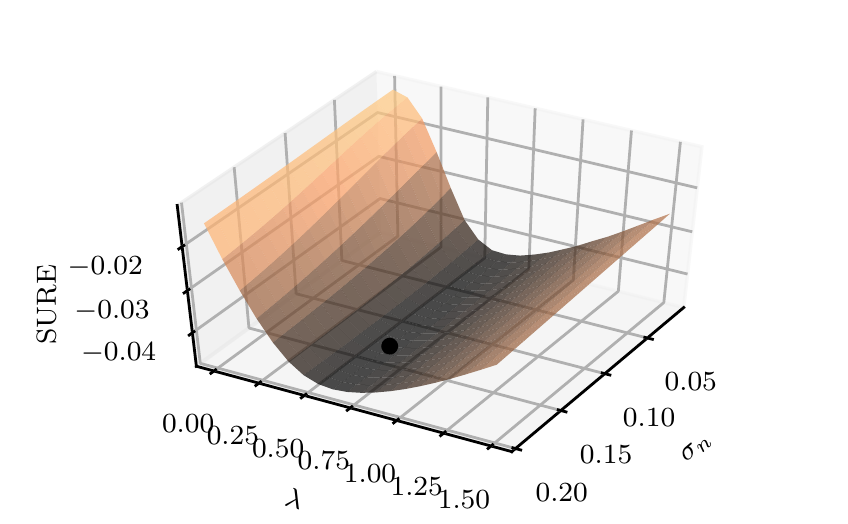}
        \caption{SURE Surface. The minimum is located at $\bullet$}
    \end{subfigure}
    \hfill
    \begin{subfigure}[t]{0.5\textwidth}
 		\includegraphics[width=1\textwidth,trim={0 0 0 20},clip]{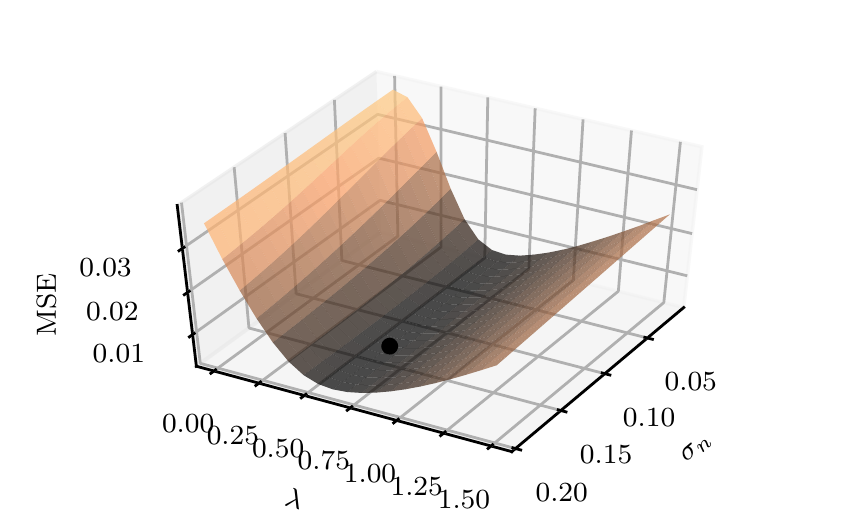}
        \caption{MSE Surface. The minimum is located at $\bullet$}
    \end{subfigure}
    \caption{SURE and MSE values for a unitary dictionary with varying $\lambda$ and $\sigma_n$.}
    \label{fig:sure}
    \begin{subfigure}[t]{0.5\textwidth}
		\includegraphics[width=1\textwidth,trim={0 7 0 5},clip]{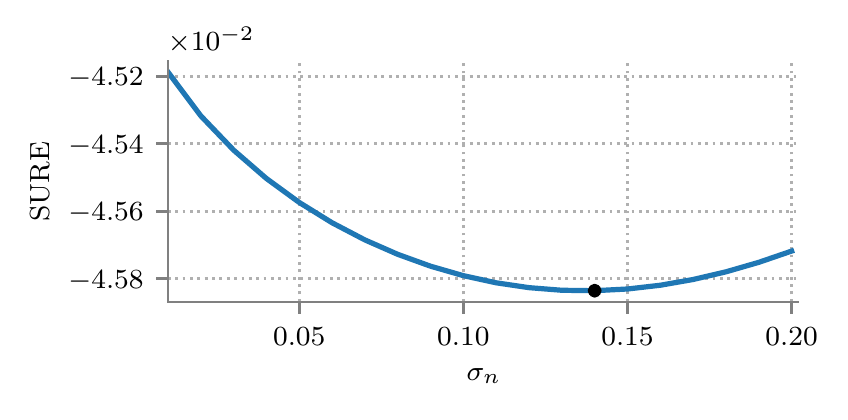}
        \caption{SURE curve for the optimal $\lambda$.}
    \end{subfigure}
    \hfill
    \begin{subfigure}[t]{0.5\textwidth}
        \includegraphics[width=1\textwidth,trim={0 10 0 5},clip]{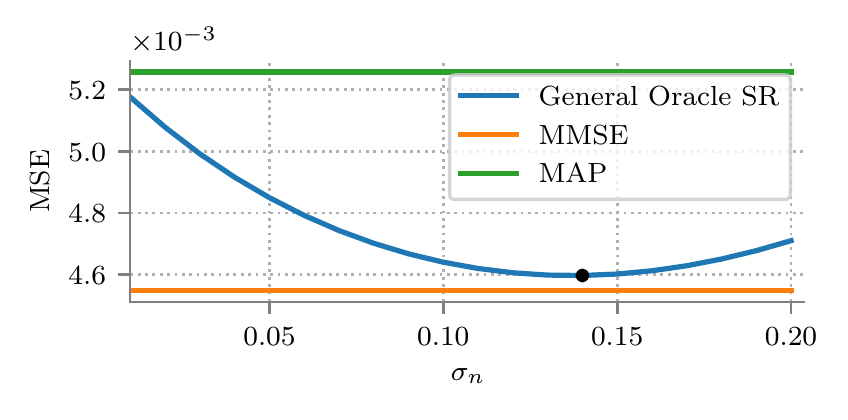}
        \caption{MSE curve for the optimal $\lambda$ extracted using SURE.}
    \end{subfigure}
    \caption{SURE and MSE values for the optimal $\lambda$, extracted from SURE. SURE's optimal $\sigma_n$ marked in $\bullet$.}
    \label{fig:sure_mse}
\end{figure*}

\section{Multiplicative Bernoulli SR Noise}
\label{appendix:benoulliNoise}
In this section we seek for an SR distribution from which we can benefit more than others in terms of computational efficiency. To do so, consider the following. Given the signal $\y$, we define the subsampling noise $\n_{\text{subsample}}$ in the following way:
\begin{equation}
n_i=\begin{cases}
1 & \text{w.p.} \quad p \\
0 & \text{w.p.} \quad 1-p
\end{cases},
\end{equation}
and the SR samples will now follow the following distribution:
\begin{equation}
y_i \cdot n_i = \begin{cases}
y_i & \text{w.p.} \quad p \\
0 & \text{w.p.} \quad 1-p
\end{cases}.
\end{equation}
Therefore, for an input signal of size $n$, only $pn$ samples will remain on average. This distribution is interesting because of the following reason. When zeroing out an element in the vector $\y$, the matching row in the dictionary $\Dv$ will always be multiplied by the zero element when calculating the correlations $\Dv^T\y_{\text{SR}}$ as done in most pursuits. This multiplication obviously has no contribution to the inner product and we might as well omit the zero elements from $\y_{\text{SR}}$ and the corresponding rows from $\Dv$, leading to a subsampled version of the signal $\y$ and the dictionary $\Dv$. In other words, in each of the SR iterations we simply subsample random elements with probability $p$ from the signal $\y$ and the matching rows from the dictionary $\Dv$, which leads to $\y_{\text{subsample}}$ of size ${pn \times 1}$ (on average) and a dictionary $\Dv_{\text{subsample}}$ of size ${pn \times m}$ (on average). Finally we sparse code the subsampled vectors. Just like in previous cases, we use the pursuit's result only as a support estimator in order to compute the oracle estimator. Hence, we then revert to the full sized signal $\y$ and dictionary $\Dv$ and compute either oracle or LS. 

Note that when using the Bernoulli noise, each pursuit has a computational benefit over the previously presented additive SR noise due to the decreased size of the signal's dimension. In Figure \ref{fig:subsampling} we show the results of the multiplicative Bernoulli noise compared to Gaussian additive noise. In this figure the $x$ axis represents the probability $p$ of the Bernoulli noise. We see that the two noise distributions lead to similar MSE results for the optimal choice of $\sigma_n$ and $p$, while using multiplicative Bernoulli SR noise is computationally efficient compared to additive Gaussian SR noise.

\begin{figure}[t]
	\includegraphics{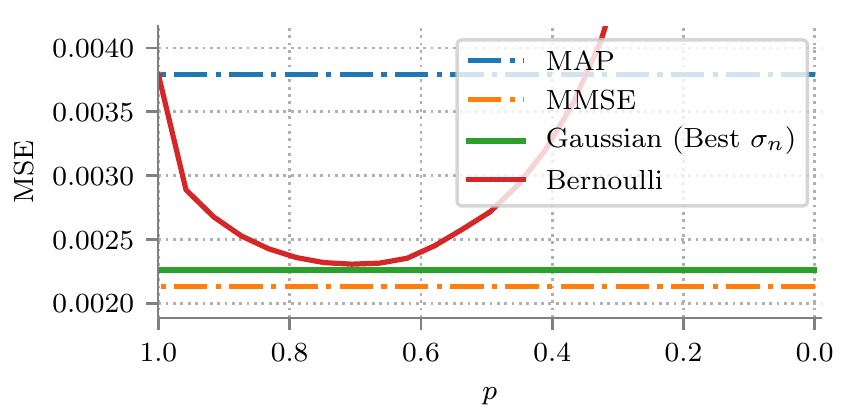}
	\caption{Multiplicative Bernoulli SR noise vs. additive Gaussian SR noise with 100 iterations of Algorithm \ref{algorithm:SR_importance}. ${\Dv \in \mathbb{R}^{50 \times 100}}$, $\nuv\sim\mathcal{N}\left(\0,\sigma_{\nuv}^{2}I\right)$, $\sigma_{\nuv}=0.2$, $\left|\left|\alphav\right|\right|_{0}=1$ and $\alpha_{S}\sim\mathcal{N}\left(0,1\right)$.}
	\label{fig:subsampling}
\end{figure}

\section*{Acknowledgment}
The authors thank Prof. Pramod K. Varshney, for his inspiring keynote talk at ICASSP 2015 in Brisbane, which inspired the initial ideas of this paper.
The research leading to these results has received funding
from the European Research Council under European Unions
Seventh Framework Programme, ERC Grant agreement no.
320649 and the Israel Science Foundation (ISF) Grant no. 335/18.

\ifCLASSOPTIONcaptionsoff
  \newpage
\fi



%
\bibliographystyle{ieeetr}
\bibliography{cites.bib}

%








\end{document}